\documentclass[letterpaper,USenglish,11pt]{paper}

\RequirePackage[pdfborder={0 0 0}]{hyperref}

\usepackage{fullpage}

\newcommand{\eps}{\ensuremath{\varepsilon}}
\usepackage[utf8x]{inputenc}

\usepackage{amsmath,amsthm}
\usepackage{tikz}
\usetikzlibrary{calc}
\usetikzlibrary{arrows}
\usepackage{rotating}
\usepackage{xspace}
\usepackage{array}
\usepackage{amsmath}
\usepackage{amsfonts}
\usepackage{amssymb}
\usepackage{etex}
\usepackage{graphicx}
\usepackage{paralist}
\usepackage{stmaryrd}
\usepackage{varioref}
\usepackage{microtype}

\usepackage{youngtab} 
\usepackage{authblk}
\usepackage{kbordermatrix}

\usepackage{listings}

\usetikzlibrary{calc}
\usetikzlibrary{arrows}
\usetikzlibrary{arrows.meta}

\usepackage{color}

\usetikzlibrary{calc}
\usetikzlibrary{decorations.pathmorphing}
\usetikzlibrary{positioning}

\usepackage{caption}
\usepackage{subcaption}

\newtheorem{thm}{Theorem}[section]
\newtheorem{prop}[thm]{Proposition}
\newtheorem{defn}[thm]{Definition}
\newtheorem{lem}[thm]{Lemma}
\newtheorem{cor}[thm]{Corollary}

\newtheorem{fact}[thm]{Fact}
\newtheorem{claim}{Claim}

\newcommand{\N}{\mathbb{N}}
\newcommand{\Z}{\mathbb{Z}}
\newcommand{\Q}{\mathbb{Q}}
\newcommand{\R}{\ensuremath{\mathbb{R}}\xspace}

\DeclareMathOperator{\rank}{rank}

\newcommand{\cB}{\ensuremath{\mathcal{B}}\xspace}
\newcommand{\cP}{\ensuremath{\mathcal{P}}\xspace}

\newcommand{\MC}{\ensuremath{\mathbf{M}}\xspace}
\newcommand{\MCB}{\ensuremath{\mathbf{M'}}\xspace}
\newcommand{\HC}{\ensuremath{\mathbf{H}}\xspace}
\newcommand{\FR}{\ensuremath{\mathbf{F}}\xspace}
\newcommand{\connmatrix}{\mathbf{C}}

\newcommand{\cupdot}{\mathbin{\mathaccent\cdot\cup}}

\newcommand{\tw}{\ensuremath{\mathrm{tw}}\xspace}
\newcommand{\pw}{\ensuremath{\mathrm{pw}}\xspace}
\newcommand{\rk}{\ensuremath{\mbox{rk}}\xspace}

\newcommand{\pathdecomp}{\mathbb{P}\xspace}

\newcommand{\executeiffilenewer}[3]{%
\ifnum\pdfstrcmp{\pdffilemoddate{#1}}%
{\pdffilemoddate{#2}}>0%
{\immediate\write18{#3}}\fi%
} 

\newcommand{\svg}[2]{\def\svgwidth{#1}%
\executeiffilenewer{#2.svg}{#2.pdf}%
{inkscape -z -D --file=#2.svg %
--export-pdf=#2.pdf --export-latex}%
{\input{#2.pdf_tex}}}

\begin{document}
\title{\centering A Tight Lower Bound for Counting Hamiltonian Cycles\\via Matrix Rank\thanks{Part of this work was done while Radu and Jesper were visiting the Simons Institute for the Theory of Computing.}} 

\author{Radu Curticapean\thanks{Institute for Computer Science and Control, Hungarian Academy of Sciences (MTA SZTAKI). \texttt{radu.curticapean@gmail.com}. Supported by ERC Grant PARAMTIGHT (No. 280152) and ERC Grant SYSTEMATICGRAPH (No. 725978).} \quad Nathan Lindzey\thanks{Department of Combinatorics and Optimization, University of Waterloo. \texttt{nlindzey@uwaterloo.ca}.}\quad Jesper Nederlof\thanks{Department of Mathematics and Computer Science, Technische Universiteit Eindhoven. \texttt{j.nederlof@tue.nl}. Supported by NWO Veni project 639.021.438}
}

\maketitle
\begin{abstract}
	For even $k \in \mathbb N$, the matchings connectivity matrix $\MC_k$ is a binary matrix indexed by perfect matchings on $k$ vertices; the entry at $(M,M')$ is $1$ iff $M \cup M'$ forms a single cycle. 
Cygan et al.~(STOC 2013) showed that the rank of $\MC_k$ over $\mathbb Z_2$ is $\Theta(\sqrt 2^k)$ and used this to give an $O^*((2+\sqrt{2})^{\mathsf{pw}})$ time algorithm for counting Hamiltonian cycles modulo $2$ on graphs of pathwidth $\mathsf{pw}$. The algorithm carries over to the decision problem via witness isolation.
The same authors complemented their algorithm by an essentially tight lower bound under the Strong Exponential Time Hypothesis (SETH). This bound crucially relied on a large permutation submatrix within $\MC_k$, which enabled a ``pattern propagation'' commonly used in previous related lower bounds, as initiated by Lokshtanov et al.~(SODA 2011).

We present a new technique for a similar pattern propagation when only a \emph{black-box} lower bound on the asymptotic rank of $\MC_k$ is given; no stronger structural insights such as the existence of large permutation submatrices in $\MC_k$ are needed. Given appropriate rank bounds, our technique yields lower bounds for counting Hamiltonian cycles (also modulo fixed primes $p$) parameterized by pathwidth.

To apply this technique, we prove that the rank of $\MC_k$ over the rationals is $4^k / \mathrm{poly}(k)$, using the representation theory of the symmetric group and various insights from algebraic combinatorics. 
We also show that the rank of $\MC_k$ over $\Z_p$ is $\Omega(1.97^k)$ for any prime $p\neq 2$ and even $\Omega(2.15^k)$ for some primes.

Combining our rank bounds with the new pattern propagation technique, we show that Hamiltonian cycles cannot be counted in time $O^*((6-\eps)^{\mathsf{pw}})$ for any $\eps>0$ unless SETH fails.
This bound is tight due to a $O^*(6^{\mathsf{pw}})$ time algorithm by Bodlaender et al.~(ICALP 2013). 
Under SETH, we also obtain that Hamiltonian cycles cannot be counted modulo primes $p\neq 2$ in time $O^*(3.97^\mathsf{pw})$ and, for some primes, not even in time $O^*(4.15^\mathsf{pw})$,
indicating that the modulus can affect the complexity in intricate ways.
\end{abstract}

\clearpage

\newcommand{\HCP}{\textsc{HC}\xspace}
\newcommand{\CHCP}{\#\textsc{HC}\xspace}
\section{Introduction}
Rank is a fundamental concept in linear algebra
and has numerous applications in diverse areas of discrete mathematics and theoretical computer science,
such as algebraic complexity~\cite{DBLP:books/daglib/0090316}, 
communication complexity~\cite{DBLP:conf/focs/LovaszS88}, 
and extremal combinatorics~\cite{matouvsek2010thirty},
to name only a few. 
A common phenomenon is that low rank often 
helps in proving combinatorial upper bounds or designing algorithms, e.g., through
representative sets~\cite{DBLP:journals/iandc/BodlaenderCKN15, DBLP:journals/jacm/FominLPS16, DBLP:journals/talg/KratschW14} or
the polynomial method (which ultimately relies on fast rectangular matrix multiplication, enabled through low-rank factorizations of problem-related matrices~\cite{DBLP:conf/fsttcs/Williams14}).
In particular, rank has recently found applications in \emph{fine-grained complexity} (see~\cite{Alman17} and the references therein) and \emph{parameterized complexity}. In the latter, several influential results, such as algorithms for kernelization~\cite{DBLP:journals/talg/KratschW14}, the longest path problem~\cite{Monien1985}, and connectivity problems parameterized by treewidth~\cite{DBLP:conf/focs/CyganNPPRW11,CyganKN13,DBLP:journals/iandc/BodlaenderCKN15}, rely crucially on low-rank factorizations.

In view of the utility of low rank in proving upper bounds,
it is natural to ask, conversely, whether high rank translates into lower bounds.
Indeed, examples for this connection can be found in communication complexity~\cite[Section 1.4]{Kushilevitz:1996:CC:264772} and circuit complexity~\cite{DBLP:conf/mfcs/1977}.
In the present paper, we find such applications also in fine-grained and parameterized complexity: 
We develop a technique that allows us to transform rank lower bounds into conditional lower bounds for the problem \CHCP of counting Hamiltonian cycles.
The decision version \HCP of \CHCP, which asks for the \emph{existence} of a Hamiltonian cycle, is a classical subject of algorithmic research.
For decades, the well-known $O^*(2^n)$ time dynamic programming algorithm~\cite{held1962dynamic} was essentially the fastest known algorithm for \HCP,
until a breakthrough result~\cite{DBLP:journals/siamcomp/Bjorklund14} showed that \HCP can actually be solved in $O^*(1.657^n)$ randomized time.
This result spawned several novel algorithmic insights into \HCP,
but also showed that we still do not understand this problem in a satisfactory way:
No \emph{deterministic} $O^*((2-\eps)^n)$ time algorithm for \HCP is known,
and even no randomized $O^*((2-\eps)^n)$ time algorithms are known for the more general traveling salesman problem, the directed Hamiltonian cycle problem, or the counting version \CHCP. 

One of the novel algorithmic techniques for \HCP
following in the wake of~\cite{DBLP:journals/siamcomp/Bjorklund14}
is closely tied to the rank of the so-called \emph{matchings connectivity matrix}~\cite{CyganKN13}.
For even $k$, the matchings connectivity matrix $\MC_k$ is indexed by the perfect matchings of the complete graph $K_k$,
and the entry $\MC_k[M,M']$ for perfect matchings $M$ and $M'$ is defined as $1$ if the union $M \cup M'$ is a single cycle, and $0$ otherwise.
See Figure~\ref{fig:mcm6} for an example.
The matchings connectivity matrix can be seen as a description of the behavior of Hamiltonian cycles under graph separators, an interpretation that proved useful for algorithmic applications.
For instance, the authors of~\cite{CyganKN13} show that the rank of $\MC_k$ over $\Z_2$ is precisely $2^{k/2-1}$
and use this surprisingly low rank to count Hamiltonian cycles modulo $2$ in bipartite directed graphs in $O(1.888^n)$ time, 
which was recently improved to $O^*(3^{n/2})$ time in~\cite{DBLP:conf/icalp/BjorklundKK17}. 
A randomized algorithm for the decision version follows from witness isolation.

\begin{figure}
	\begin{center}
		\includegraphics[width=10cm]{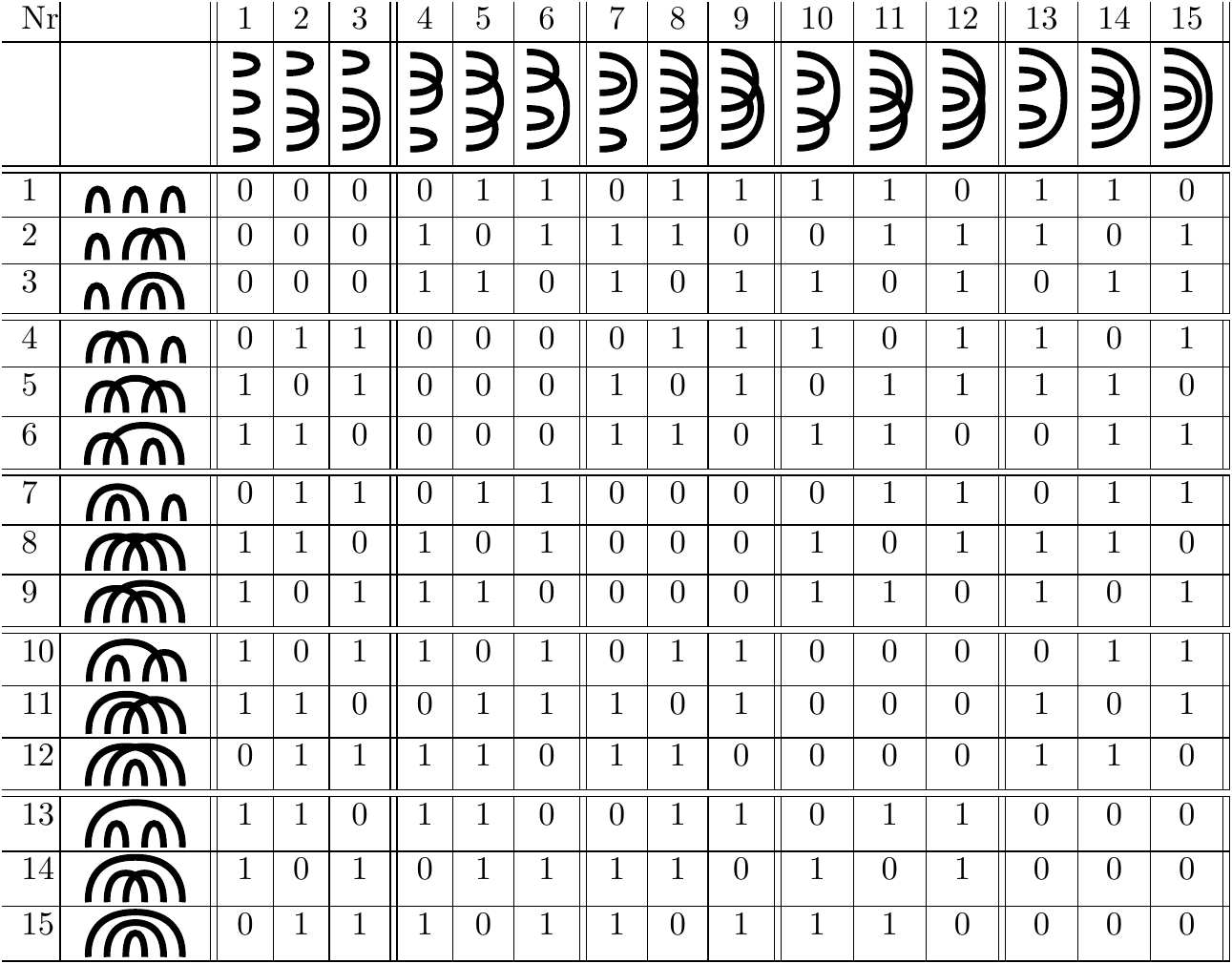}
	\end{center}
	\caption{Originally from~\cite{CyganKN13}, this figure displays the matchings connectivity matrix $\MC_6$, 
	which indicates which pairs of the $15$ perfect matchings on $6$ vertices form a Hamiltonian cycle. }
	\label{fig:mcm6}
\end{figure}

The low rank of $\MC_k$ also enabled an $O^*((2+\sqrt{2})^{\pw})$ time algorithm for \HCP on graphs with a given path decomposition of width $\pw$.
For this problem, the standard dynamic programming approach would require to keep track of all partitions of separators, resulting in a running time of $O^*(2^{\pw \log \pw})$;
it is thus somewhat remarkable that the single-exponential running time of $O^*((2+\sqrt{2})^{\pw})$ can be achieved.
Even more surprisingly, the base $2+\sqrt{2}$ appears to be optimal,
as it is known that any $O((2+\sqrt{2}-\eps)^{\pw})$ time algorithm would violate the Strong Exponential Time Hypothesis (SETH)~\cite{CyganKN13}. This was proven by combining a general reduction technique for SETH-based lower bounds from~\cite{DBLP:conf/soda/LokshtanovMS11a} with a special property of $\MC_k$, namely, that $\MC_k$ contains a principal minor of size $2^{k/2-1}$ that is a permutation matrix. In other words, there is a collection of $2^{k/2-1}$ perfect matchings such that every perfect matching in this collection can be extended to a Hamiltonian cycle by precisely one other member.

The general technique for SETH-based lower bounds from~\cite{DBLP:conf/soda/LokshtanovMS11a} 
was successfully applied to various problems parameterized by pathwidth:
As a result, the optimal base in the exponential dependence on the pathwidth has been identified for many problems, assuming SETH.
However, there are still natural open problems left, such as \CHCP: It has been shown that this problem can be solved in $O^*(6^{\pw})$ time~\cite{DBLP:journals/iandc/BodlaenderCKN15}, later extended to $O^*((2^\omega+2)^{\tw})$ time when parameterized by treewidth~\cite{DBLP:conf/iwpec/Wlodarczyk16}, where $2\leq\omega< 2.371$ denotes the matrix multiplication constant. A tight lower bound however remained elusive, and this might justify optimism towards improved algorithms:
For example, if we could lift the $O^*((2+\sqrt{2})^{\pw})$ time algorithm for \CHCP modulo $2$ to an $O^*((4-\eps)^{\pw})$ time algorithm for \CHCP, we could solve \CHCP on bipartite graphs in $O^*((2-\eps')^n)$ time, since bipartite graphs have pathwidth at most $n/2$. 

\subsection*{Our main results}

We strike out the route towards faster algorithms for \CHCP sketched above: We show that the current pathwidth (and, assuming $\omega=2$, treewidth) based algorithms are optimal assuming SETH.
\begin{thm}\label{thm:mainlb}
	Assuming SETH, there is no $\eps >0$ such that \CHCP can be solved in $O^*((6-\eps)^{\pw})$ time
	on graphs with a given path decomposition of width $\pw$.
\end{thm}
This theorem gives a natural example for an NP-hard problem whose decision version (with base $2+\sqrt 2$) and counting version (with base $6$) differ provably under SETH.

We prove Theorem~\ref{thm:mainlb} by starting from the general reduction technique in~\cite{DBLP:conf/soda/LokshtanovMS11a}, 
augmented with a novel idea:
We extend the technique in such a way that it can exploit arbitrary lower bounds on the matrix rank of $\MC_k$,
without further insights into the particular structure of basis vectors.
That is, we derive Theorem~\ref{thm:mainlb} as a consequence of the following more general ``black-box'' connection 
between the rank of $\MC_k$ and the running time for \CHCP:
If the exponential base of the rank can be lower-bounded by $r$, 
then we do not expect $O^*((2+r -\eps)^{\pw})$ time algorithms.
\begin{thm}\label{thm:genlowhc}
	Let $r \in \R$ be such that $\log_r(\rank(\MC_k))/k \to c$, where $c \geq 1$, as even $k$ tends to infinity.\footnote{This implies that $\rank(\MC_k)$ can be lower-bounded by $\Omega(r^k)$ up to sub-exponential factors.}
	Assuming SETH, there is no $\eps >0$ such that \CHCP can be solved in $O^*((2+r -\eps)^{\pw})$ time on graphs with a given path decomposition of width $\pw$.

	For prime numbers $p$, the same applies to \CHCP modulo $p$ when replacing $r$ by $r_p$, 
	which is defined analogously to $r$ by taking $\rank(\MC_k)$ over $\Z_p$.
\end{thm}
To prove Theorem~\ref{thm:mainlb}, we then combine Theorem~\ref{thm:genlowhc} with our second main contribution:
We determine the rank of $\MC_k$ over $\Q$ up to polynomial factors,
and for primes $p\neq 2$, we additionally give lower bounds on the rank over $\Z_p$ that are higher than the rank over $\Z_2$.
\begin{thm}\label{thm:mainrank}
	The rank of $\MC_k$ over the rational numbers is at least $\Omega(4^k / k^3)$. 
	For any prime $p\neq 2$, the rank of $\MC_k$ over $\Z_p$ is at least $\Omega(1.979^k)$,
	and for prime $5 \leq p \leq 13$, it is at least $\Omega(2.152^k)$. 
	See Theorem~\ref{thm:modranks-all} for a full list.
\end{thm}
The bound over $\Q$ is obtained by a novel application of representation theory,
inspired by a previous approach from~\cite{RazS95},
where the rank of a bipartite version of $\MC_k$ over $\Q$ was found to be $\Theta(2^k)$ up to polynomial factors.
In the bipartite version of $\MC_k$, only perfect matchings contained in the complete bipartite graph $K_{k/2,k/2}$ are considered.
In our non-bipartite version, any perfect matching in the complete graph $K_k$ is allowed;
this is more appropriate for algorithmic applications,
and our new bound over $\Q$ shows that going to the non-bipartite setting increases the rank significantly.

Combined with Theorem~\ref{thm:genlowhc}, our bound over $\Z_p$ suggests that \CHCP modulo prime $p \neq 2$ is harder than modulo $2$:
We can solve \CHCP modulo $2$ in $O^*((2+\sqrt{2})^{\pw})$ time, where $\sqrt{2}\leq 1.42$, but we cannot solve \CHCP modulo $p\neq 2$ in $O^*((2+1.97)^\pw)$ time unless SETH fails.
This connects to recent results~\cite{DBLP:conf/focs/BjorklundH13,DBLP:conf/icalp/BjorklundKK17}, which show that the counting Hamiltonian cycles modulo $c^n$ (not parameterized by pathwidth) can be solved in time $O^*((2-\eps_c)^n)$, where $\eps_c >0$ depends on the constant $c$.

\subsection*{Connection matrices and fingerprints} 
The matchings connectivity matrix $\MC_k$ fits into a bigger picture of so-called \emph{connection matrices} for graph parameters,
and our bounds on the rank of $\MC_k$ translate into rank bounds in this framework. 

The connection matrices of a graph parameter $f$ are a sequence of matrices $\connmatrix_{k}$, for $k\in \N$, which describe the behavior of $f$ under graph separators of size $k$.
To define these matrices, say that a \emph{$k$-boundaried graph}, for $k \in \mathbb N$, is a simple graph with $k$ distinguished vertices that are labeled $1,\ldots,k$.
Two $k$-boundaried graphs $G$ and $H$ can be glued together, yielding a graph $G \oplus H$, by taking the disjoint union of $G$ and $H$ and identifying vertices with the same label.
The $k$-th connection matrix $\connmatrix_{k}$ of $f$ then is an infinite matrix 
whose rows and columns are indexed by $k$-boundaried graphs
such that the entry $\connmatrix_{k}[G,H]$ is $f(G \oplus H)$.

The ranks of connection matrices are closely related to graph-theoretic, algorithmic, and model-theoretic properties of graph parameters \cite{DBLP:conf/csl/KotekM12,DBLP:books/daglib/0031021,DBLP:journals/ejc/Lovasz06}.
In particular, the connection matrices $\connmatrix_{k}$ for the number of Hamiltonian cycles were studied in \cite{DBLP:journals/ejc/Lovasz06,DBLP:books/daglib/0031021}, where their rank was upper-bounded by $2^{O(k \log k)}$. 
As a consequence of Theorem~\ref{thm:mainrank}, we can improve upon this and obtain the following essentially tight bounds.

\begin{thm}
\label{thm:connection-rank}
For $k\in \N$, the rank of the connection matrix $\connmatrix_{k}$ for the number of Hamiltonian cycles is $6^k$, up to polynomial factors.
\end{thm}

To prove this theorem, we use a third matrix, the fingerprint matrix $\HC_k$ for Hamiltonian cycles,
which will also play an important role in our main reduction.\footnote{To avoid (or add) confusion, let us stress that we consider three matrices related to the Hamiltonian cycle problem:
The matchings connectivity matrix $\MC_k$,
the connection matrix $\connmatrix_{k}$ for the number of Hamiltonian cycles,
and the fingerprint matrix $\HC_k$ for Hamiltonian cycles.
We will revisit their differences in Section~\ref{sec:prel}.
While these matrices are closely related, our arguments benefit from using different matrices for different proofs.}
A \emph{fingerprint} of a $k$-boundaried graph is a pair $(d,M)$, where $d \in \{0,1,2\}^k$ assigns $0$, $1$ or $2$ to each boundary vertex, and $M$ is a perfect matching on the boundary vertices to which $d$ assigns $1$.
Fingerprints are essentially the states one would use in the natural dynamic programming routine for counting Hamiltonian cycles parameterized by pathwidth; they describe the behavior of a Hamiltonian cycle on a given side of a separation.
A pair of fingerprints $(d,M)$ and $(d',M')$ on $B$ \emph{combines} if $d_v+d'_v=2$ for every $v \in B$ and additionally $M \cup M'$ forms a single cycle.
The fingerprint matrix $\HC_k$ is a binary matrix, indexed by fingerprints, 
and the value at a pair of fingerprints is $1$ iff the two fingerprints combine.

It can be derived easily from our rank bound for the matchings connectivity matrix $\MC_k$ 
that the rank of $\HC_k$ is $6^k$ up to polynomial factors, see Fact~\ref{fact:rank-MC-to-HC}.
To establish Theorem~\ref{thm:connection-rank},
we show in Fact~\ref{fact:rank-HC-to-connectionmatrix} 
that $\connmatrix_{k}$ and $\HC_k$ have the same rank.

\subsection*{Proof techniques}

In the remainder of the introduction, 
we sketch the techniques used to obtain Theorems~\ref{thm:mainrank} and \ref{thm:genlowhc},
which together imply Theorem~\ref{thm:mainlb}.

\subsubsection*{Theorem~\ref{thm:mainrank}: Rank of the matchings connectivity matrix}

To prove Theorem~\ref{thm:mainrank}, we give two different lower bounds on the rank of $\MC_k$:
One is relatively simple and contained in Section~\ref{sec: modrank}.
For this bound, we first explicitly compute the rank of small matching connectivity matrices
and then use a product construction to give lower bounds for larger orders.
While the resulting bound is loose,
it also applies to the rank of $\MC_k$ over $\mathbb{Z}_{p}$ for prime $p\neq2$, whereas our more sophisticated main bound does not.
In particular, we can use the bound to show that the rank of $\MC_k$ over $\mathbb{Z}_{3}$ and other primes is asymptotically larger than the rank over $\mathbb{Z}_{2}$.

Our main result however concerns the rank of $\MC_k$ over $\Q$, 
which we establish to be $4^k$ up to polynomial factors in Section~\ref{sec: reprank}.
To this end, we build upon representation-theoretic techniques that were also used in Raz and Spieker's bound~\cite{RazS95} for the bipartite version of $\MC_k$, and which we first survey briefly:
A \emph{hook partition} $\lambda$ of some number $k\in \N$ is a number partition 
with the particular form $(t,1,\ldots,1)$ for some $t\leq k$. 
One can view $\lambda$ as a \emph{Ferrers diagram},
which is a left-adjusted diagram made of cells,
as shown in Figure~\ref{fig:ferrer}.
A \emph{standard Young tableau} of shape $\lambda$ is a labeling of this diagram with numbers from $[k]$
such that the numbers are strictly increasing in each row and each column, see Figure~\ref{fig:young}.
Raz and Spieker showed that the rank of the bipartite variant of $\MC_{2k}$ can be expressed as a weighted sum over all hook partitions $\lambda$ of $k$, where each $\lambda$ is weighted by the squared number of Young tableaux of shape $\lambda$. 
This sum simplifies to the central binomial coefficient ${2k-2 \choose k-1}$, 
showing that the bipartite variant of $\MC_{k}$ has rank $\Theta(2^k)$, up to polynomial factors.

To address the non-bipartite setting, we found ourselves in need of additional techniques from algebraic combinatorics that were not present in Raz and Spieker's original bound,
such as the perfect matching association scheme \cite{GodsilMeagher}
and zonal spherical functions~\cite{MacDonald95}.
With these at hand, we prove that the rank of $\MC_{2k}$ can be lower-bounded by a similar sum over number partitions $\lambda$ as in the bipartite case,
this time however ranging over \emph{domino hook partitions} $\lambda$, which have the form $(2t, 2t, 2, \ldots , 2)$ for some $t \leq k$.
As in the bipartite case, we then observe that this sum simplifies significantly, this time however to (essentially) a product of two consecutive Catalan numbers. This entails a lower bound of $4^k$ for the rank of $\MC_{k}$, up to polynomial factors. It then follows easily from the upper bound in the bipartite setting that this bound is tight up to polynomial factors. 

\subsubsection*{Theorem~\ref{thm:genlowhc}: SETH-hardness via assignment propagation}
\label{sub: lms-reduction}
To describe how we turn lower bounds on the rank of $\MC_{k}$ into algorithmic intractability results for \CHCP under SETH,
let us first survey the general construction from~\cite{DBLP:conf/soda/LokshtanovMS11a},
which we dub a \emph{block propagation scheme}:
Given a CNF-formula $\varphi$ with $n$ variables, such a scheme produces an equivalent instance $I$ of the target problem with parameter value $k \leq cn$ for some constant $c \leq 1$.
An algorithm with running time $O^*((2^{1/c}-\eps)^k)$ for the target problem would then refute SETH, 
as it would imply a $O^*((2-\eps)^n)$ time algorithm for CNF-SAT.

The constructed target instance $I$ has the outline sketched in Figure~\ref{fig:rectprop}:
\begin{figure}
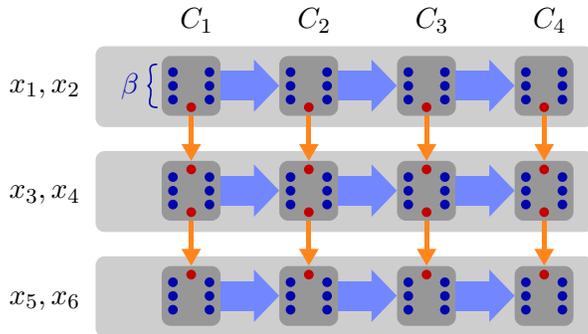

\begin{centering}
\svg{8cm}{block-propagation-scheme}
\par\end{centering}
\caption{\label{fig:rectprop}A block propagation scheme for 
$n=6$ variables, 
$m=4$ clauses,
variable block size $\gamma = 2$, 
and row pathwidth $\beta = 3$. 
The $q = n / \gamma = 3$ variable blocks correspond to rows.
Assignments to blocks are propagated from left to right.
Clauses correspond to columns
and clause satisfaction is propagated downwards.}
\end{figure}
The $n$ variables of $\varphi$ are grouped into $q = \lceil n/\gamma \rceil$ blocks of constant size $\gamma \in \N$, 
where $\gamma$ depends only on the $\epsilon$ in the running time we wish to rule out.
The variable blocks are represented as rows, 
each propagating an assignment of type $\{0,1\}^\gamma$
using a thin graph of pathwidth $\beta \in \N$.
Specifically, an assignment $\{0,1\}^\gamma$ is represented as the type of a partial solution for the target problem on a $\beta$-boundaried graph (e.g., as a partial coloring of the boundary, or in our case, as a fingerprint of a Hamiltonian cycle.)
The relationship between $\beta$ and $\gamma$ is important in this construction:
Intuitively, if we can choose $\beta$ small for large $\gamma$, then the target problem has a large ``combinatorial capacity'' in the sense that it allows us to pack assignments to large blocks into thin wires.

In a block propagation scheme, the clauses of $\varphi$ are then represented as columns;
the column corresponding to clause $C$ checks whether the overall assignment of type $\{0,1\}^n$ propagated by the rows satisfies $C$.
To this end, one can use \emph{cell gadgets},
which are graphs with $\beta$ left and $\beta$ right interface vertices, and $c$ additional top/bottom interface vertices,
where $c \in \N$ depends only on the target problem.
The cell gadget is placed at the intersection of a row and a column,
and it needs to ``decode'' an assignment $x \in \{0,1\}^\gamma$ from the state of the left $\beta$ interface vertices,
decide whether $x$ satisfies the clause,
and ``encode'' $x$ back into the state of right $\beta$ vertices.
The top and bottom interface is used to propagate, from the top of a column downwards, 
whether the respective clause is already satisfied by the partial assignments to the blocks above a given cell.
Due to the grid-like construction, the overall pathwidth of the instance $I$ is usually easily seen to be bounded by $\beta q + O(1)$, where the additive constant accounts for the size of cell gadgets.

The main technical effort in these reductions lies in constructing the cell gadget,
and this usually subsumes constructing a ``state tester'', 
a gadget that tests whether, in a solution to $I$, the $\beta$ left/right interface vertices are in a particular state $S$ (say, a particular partial coloring, or a particular fingerprint of a Hamiltonian cycle).
This requires constructing a graph that can be extended to a solution of the target instance iff the $\beta$ relevant vertices are in state $S$,
and for various problems, such constructions can be achieved with some effort.
In the case of \CHCP,
we face the situation that testers for fingerprints of Hamiltonian cycles \emph{do not exist}:
There are fingerprints $S$ such that any graph that extends $S$ to a Hamiltonian cycle 
also extends some unwanted fingerprints $S' \neq S$.
Our main insight here is that this problem can be solved by firstly restricting to a set of $6^\beta$ good fingerprints 
that induce a full-rank submatrix $\mathbf{F}$ of the fingerprint matrix $\HC_\beta$, 
and secondly simulating a ``linear combination'' of testers, 
with coefficients obtained from the inverse of $\mathbf{F}$.
In the fingerprint tester for $S$, other fingerprints $S'\neq S$ \emph{will} have extensions of non-zero weight,
but the weights of these extensions are chosen in such a way (depending on $\mathbf{F}^{-1}$) that extensions of $S' \neq S$ cancel out. (A similar idea was used before to obtain conditional lower bounds for the complexity of permanents \cite{DBLP:conf/focs/CurticapeanX15}.)
This allows us to simulate a state tester for fingerprints,
and the use of cancellations also shows why this lower bound works only for \CHCP and not for \HCP
---which is fortunate, since \HCP does admit an $O^*((2+\sqrt2)^{\pw})$ time algorithm.
\section{Preliminaries and notation}\label{sec:prel}
If $\varphi$ is a CNF-formula, and $x$ is a (partial) assignment to its variables, 
we write $x \models \varphi$ to denote that $x$ satisfies $\varphi$.
For integers $n$, we let $[n]=\{1,\ldots,n\}$.
All graphs in this paper will be undirected.
If $G=(V,E)$ is a graph, $v \in V$ is a vertex and $X \subseteq E$ is an edge set, we let $d_X(v)$ denote the number of edges in $X$ that are incident to $v$.
We write $\mathcal{M}_{2n}$ for the set of all perfect matchings of the complete graph $K_{2n}$.

\paragraph{Strong Exponential Time Hypothesis:} 
As formulated by Impagliazzo and Paturi~\cite{ImpagliazzoP01}, the complexity assumption SETH states that for every $\varepsilon > 0$, there is a constant~$k$ such that $k$-CNF-SAT (the satisfiability problem for $k$-CNF formulas on $n$ variables) cannot be solved in time $O^*((2-\eps)^n)$. It is common to state this hypothesis as ruling out even randomized algorithms, making it slightly stronger. A result from Calabro et al.~\cite[Theorem 1]{DBLP:journals/jcss/CalabroIKP08} gives a randomized reduction from $k$-CNF-SAT to the problem UNIQUE-$k$-CNF-SAT, where the $k$-CNF formula is guaranteed to have at most one satisfying assignment.
This allows us to assume that the $k$-CNF formulas in the statement of (randomized) SETH have at most one satisfying assignment.

\paragraph{Pathwidth:} A \emph{path decomposition} of a graph~$G=(V,E)$ is a path~$\pathdecomp$ in which each node~$x$ has an associated set of vertices~$B_x\subseteq V$ (called a \emph{bag}) such that~$\bigcup_x B_x=V$ and the following holds:

\begin{enumerate}
\item For each edge~$\{u,v\}\in E$ there is a node~$x$ in~$\pathdecomp$ such that~$u,v\in B_x$.
\item If $v\in B_x\cap B_y$ then~$v\in B_z$ for all nodes~$z$ on the (unique) path from~$x$ to~$y$ in~$\pathdecomp$.
\end{enumerate}

The \emph{width} of~$\pathdecomp$ is the size of the largest bag minus one, and the pathwidth of a graph~$G$ is the minimum width over all possible path decompositions of~$G$.
A path decomposition \emph{starts in $L$} if the first bag contains $L$ and \emph{ends in $R$} if the last bag contains $R$.

Since our focus here is on path decompositions, we only mention in passing that the related notion of treewidth can be defined in the same way, except for letting the nodes of the decomposition form a tree instead of a path.

\paragraph{Kronecker products:}
Given a field $\mathbb F$ and two matrices $\mathbf{A}\in \mathbb F^{n \times m}$ and $\mathbf{B}\in \mathbb F^{n' \times m'}$, 
the Kronecker product $\mathbf{A} \otimes \mathbf{B}$ 
is a matrix in $\mathbb F^{n \cdot n' \times m \cdot m'}$. 
Its rows can be indexed by pairs $(i,i') \in [n]\times[n']$, and similarly for columns.
The entry of $\mathbf{A} \otimes \mathbf{B}$ at row $(i,i')$ and column $(j,j')$ is defined as $\mathbf{A}[i,j] \cdot \mathbf{B}[i',j']$. For $t \in \N$, the $t$-th Kronecker power of $\mathbf{A}$ is the $t$-fold product $\mathbf{A}^{\otimes t} = \mathbf{A} \otimes \ldots \otimes \mathbf{A}$, and we consider its rows and columns to be indexed by $[n]^t$ and $[m]^t$ respectively.

If $\mathbf{A}$ and $\mathbf{B}$ each have full rank over $\mathbb F$, then so does $\mathbf{A} \otimes \mathbf{B}$.
Note that this requires $\mathbb F$ to be a field;
it would fail if $\mathbb F$ contained zero divisors.
For our purposes of computing the rank of matrices over $\Z_p$, 
this means we require $p$ to be prime.

\paragraph{Fingerprints:}
The \emph{HC-fingerprints} (which we often abbreviate as \emph{fingerprints}) capture the states of the natural dynamic program for Hamiltonian cycles:
\begin{defn}[Fingerprint, Partial Solutions]
	Let $G=(V,E)$ be a graph and let $B \subseteq V$, where $B$ is the set of `boundary vertices'. 
	A \emph{fingerprint on $B$} is a pair $(d,M)$ where 
	$d \in \{0,1,2\}^B$ and 
	$M$ is a perfect matching on $d^{-1}(1)$. 
	A \emph{partial solution in $G$ for $(d,M)$} is an edge set $H \subseteq E$ such that
	\begin{inparaenum}[(i)]
		\item $d_H(v)=2$ for every $v \in V \setminus B$,
		\item $d_H(v)=d_v$ for every $v \in B$, and
		\item if $(u,v) \in M$ then $u$ and $v$ are the endpoints of the same path of $H$.
	\end{inparaenum}	

	Two fingerprints $(d,M)$ and $(d',M')$ on $B$ \emph{combine (or match)} if $d_v+d'_v=2$ for every $v \in B$ and $M \cup M'$ forms a single cycle or is empty.
\end{defn}

\paragraph{Variants of connection matrices:}
Our paper studies three related matrices that describe the behavior of Hamiltonian cycles under separators. We recall their definitions here for reference:
 \begin{enumerate}
\item[$\MC_k$, the $k$-th matchings connectivity matrix:] This binary matrix is indexed by perfect matchings $M,M' \in \mathcal{M}_k$, and $\MC_k[M,M']$ is $1$ iff $M \cup M'$ is a single cycle. This matrix appears naturally in our rank lower bounds.
\item[$\HC_k$, the $k$-th fingerprint matrix (for Hamiltonian cycles):] This binary matrix is indexed by all fingerprints $(d,M)$ on a fixed set $B$ of size $k$. An entry $\HC_k[f,f']$ equals $1$ iff $f$ and $f'$ combine. This matrix will be crucially used in the algorithmic lower bound.
\item[$\connmatrix_k$, the $k$-th connection matrix (for the number of Hamiltonian cycles):] This integer-valued matrix is indexed by all $k$-boundaried graphs, and hence infinite. The entry $\connmatrix_k[G,G']$ counts the Hamiltonian cycles in the graph $G \oplus G'$. (If an edge between boundary vertices is present in both $G$ and $G'$, we count it twice in $G \oplus G'$.) This matrix will not be used in later sections, but we still mentioned it to \emph{connect} to the established literature on connection matrices.
\end{enumerate}
The subscript $k$ is omitted when clear from the context. 
As we show below, we can easily transform rank bounds for $\MC_k$ into bounds for $\HC_k$ and $\connmatrix_k$, thus justifying our focus on the rank of $\MC_k$:
\begin{fact}
\label{fact:rank-MC-to-HC}
If $p$ is a polynomial such that $\rank(\MC_k) \geq c^k / p(k)$, 
then $\rank(\HC_k) \geq (2+c)^k / q(k)$ for some polynomial $q$.
\end{fact}
\begin{proof}
Subject to proper indexing,
the matrix $\HC_k$ is a block-antidiagonal matrix that has a block for every vector $d \in \{0,1,2\}^{k}$,
since fingerprints with degree functions $d, d'$ not satisfying $d(v)+d'(v)=2$ for all $v$ cannot match. 
Therefore, we obtain:
\[
	\rank(\HC_k) 
	\ \geq \ \sum_{\substack{|d^{-1}(1)|=i \\ i \text{ even}}} \binom{k}{i} \cdot 2^{k-i} \cdot \rank(\MC_i) 
	\ \geq \ \frac{1}{k}\sum_{i=1}^k \binom{k}{i} \cdot 2^{k-i} \cdot c^i / p(i)
	\ \geq \ \frac{1}{k} (2+c)^k / q(k),
\]
where $q$ is a polynomial satisfying $p(i) \geq q(i)$ for $i=1,\ldots,k$, and the last inequality follows from the binomial theorem. The upper bound follows similarly.
\end{proof}
Furthermore, a simple argument shows that the fingerprint matrix $\HC_k$ and the connection matrix $\connmatrix_k$ actually have the same rank.
\begin{fact}
\label{fact:rank-HC-to-connectionmatrix}
For every $k \in \N$, the matrices $\HC_k$ and $\connmatrix_k$ have the same rank.
\end{fact}
\begin{proof}
We first show that $\rank(\connmatrix_k) \geq \rank(\HC_k)$
by finding $\HC_k$ as a submatrix of $\connmatrix_k$.
To this end, we construct a $k$-boundaried graph $G_F$ for every $k$-fingerprint $F$
and then find $\HC_k$ as the submatrix induced by these graphs.
Given $F$, the graph $G_F$ is constructed as follows:
At first, it contains only the boundary vertices $1,\ldots,k$.
Then we add an arbitrary partial solution for $F$ to $G_F$.
For instance, if $F=(d,M)$ and $M$ is non-empty, pick the lexicographically first edge of the matching $M$, say $ij \in M$,
and connect $i$ to $j$ in $G_F$ with a path that passes through all vertices in $d^{-1}(2)$ in an arbitrary order. Then add all edges in $M \setminus \{ij\}$ to $G_F$ as edges. If $M$ is empty, add a Hamiltonian cycle on $d^{-1}(2)$.
Finally, subdivide all edges of the graph; 
this adds some number of subdivision vertices to $G_F$, which we consider not to be part of the boundary.
Note that the degree of boundary vertex $i \in [k]$ in $G_F$ is precisely $d(i)$.

Given two $k$-fingerprints $F,F'$, we observe that any Hamiltonian cycle in $G_F \oplus G_{F'}$ uses all edges of the graph, 
as every edge is incident to a (subdivision) vertex of degree $2$. 
This implies, firstly, that the number of Hamiltonian cycles in $G_F \oplus G_{F'}$ is either $0$ or $1$.
Secondly, it implies that $d_F + d_{F'}$ needs to be the constant $2$-function for $G_F \oplus G_{F'}$ to have a Hamiltonian cycle.
If this condition is fulfilled, then by construction, $G_F \oplus G_{F'}$ has a Hamiltonian cycle iff $M \cup M'$ forms a single cycle.
Summarizing, we have that $\connmatrix_{k}[G_F , G_{F'}] \in \{0,1\}$
and that $\connmatrix_{k}[G_F , G_{F'}] > 0$ iff $F$ and $F'$ match.
This shows that the set of $k$-boundaried graphs $G_F$, for $k$-fingerprints $F$, induce the fingerprint matrix $\HC_k$ as a submatrix in $\connmatrix_k$,
and the lower bound on the rank of $\connmatrix_k$ follows.

For the upper bound of $\rank(\connmatrix_k) \leq \rank(\HC_k)$, 
we find a matrix $\mathbf{A}_k$ such that 
$\connmatrix_k = \mathbf{A}_k \cdot \HC_k \cdot \mathbf{A}_k^T$.
The rows of $\mathbf{A}_k$ are indexed by $k$-boundaried graphs $G$,
the columns are indexed by fingerprints $F$, 
and we define $\mathbf{A}_k[G,F]$ to count the partial solutions in $G$ for the fingerprint $F$.

Given two $k$-boundaried graphs $G$ and $G'$,
every Hamiltonian cycle $C$ in $G \oplus G'$ induces a partial solution in each of $G$ and $G'$, for fingerprints $F$ and $F'$, respectively. The pair of fingerprints $P_C = (F,F')$ can be determined uniquely from the partial solutions of $C$, and since $C$ is a Hamiltonian cycle, it follows that $F$ and $F'$ match. 
Given a matching pair of fingerprints $(F,F')$, the number of Hamiltonian cycles of $G \oplus G'$ with $P_C = (F,F')$ is precisely $\mathbf{A}_k[G,F] \cdot \mathbf{A}_k[G',F']$, as the extensions in each of $G$ and $G'$ can be chosen independently,
provided they agree with $F$ and $F'$ respectively.
We conclude that the number of Hamiltonian cycles in $G \oplus G'$ can be expressed as
\[
\sum_{\substack{F,F' \\ \text{match}}} {\mathbf{A}_k[G,F] \cdot \mathbf{A}_k[G',F']}
= (\mathbf{A}_k \cdot \HC_k \cdot \mathbf{A}^T_k)[G,G'].
\]
It follows that $\connmatrix_k = \mathbf{A}_k \cdot \HC_k \cdot \mathbf{A}_k^T$ as claimed, establishing the upper bound on the rank.
\end{proof}

\section{A simple rank lower bound}
\label{sec: modrank}
\newcommand{\MSTAR}{\mathbf{F}}
\newcommand{\Match}{\mathcal{M}}
Before showing our main lower bound on the rank of the matchings connectivity matrix $\MC_k$ in Section~\ref{sec: reprank},
we first establish the second part of Theorem~\ref{thm:mainrank};
this turns out to be somewhat simpler.

Let $p \neq 2$ be a fixed prime.
To obtain the lower bound on the rank of $\MC_k$ over $\Z_p$, 
we proceed in two steps: 
First, we use a computer program to compute, 
for a small constant $B \in \N$, 
the rank of $\MC_B$ over $\Z_p$.
Then we use a product construction to 
amplify this initial rank to a lower bound on the rank of $\MC_{tB}$ for $t\in \N$.

\subsection{The initial matrix}

We choose $B \in \mathbb N$ maximally such that $\MC_B$ can still be computed,
e.g., by the MATLAB script provided in the ancillary files.
If the rank of $\MC_B$ over $\mathbb{Z}_p$ is $r$,
then the symmetry of $\MC_B$ implies the existence of a set $\mathcal I$ of perfect matchings in $K_B$
such that the submatrix $\MSTAR = \MC_B[\mathcal{I},\mathcal{I}]$ has full rank over $\Z_p$.
Our computations enable the following choices:

\begin{lem}
\label{lem:initial-matrices}
For any prime $p \notin \{2,3\}$,
the matrix $\MC_{10}$ has (full) rank $945$ over $\mathbb{Z}_p$.
Furthermore, the matrix $\MC_{12}$ has rank $3618$ over $\mathbb{Z}_3$,
rank $9890$ over $\mathbb{Z}_5$, and
rank $9933$ over $\mathbb{Z}_p$ for $p \in \{7, 11, 13\}$.
\end{lem}
\begin{proof}
The dimensions of $\MC_{10}$ are $945 \times 945$.
Our calculations show that the determinant of $\MC_{10}$ is non-zero and contains only the prime factors $2$ and $3$. It follows that $\MC_{10}$ has full rank over $\mathbb{Z}_p$ for any prime $p \notin \{2,3\}$. 
Over $\mathbb{Z}_3$, the rank of $\MC_{10}$ is found to be $567$,
but we will obtain a better bound by going to $\MC_{12}$, a matrix of dimensions $10395 \times 10395$,
where the claimed rank bounds can be obtained by calculation.
\end{proof}

For concreteness, we illustrate our approach in the next subsections with $\MC_6$ as initial matrix, see Figure~\ref{fig:mcm6}, and revisit the better choices provided in Lemma~\ref{lem:initial-matrices} at the end of the proof.
Calculation shows that $\det(\MC_6) = -2^{17}$,
so $\MC_6$ has full rank over $\Z_p$ for primes $p\neq 2$,
and we can choose a set $\mathcal{I}$ of size $15$
to get a full-rank matrix $\MSTAR = \MC_B[\mathcal{I},\mathcal{I}]$.
Already $B=6$ gives a lower bound on the rank of $\MC_n$ over $\mathbb{Z}_{p}$ with $p \neq 2$ that is higher than the rank over $\mathbb{Z}_{2}$.

\subsection{Amplification via Kronecker products}

After having obtained $\MSTAR = \MC_B[\mathcal{I},\mathcal{I}]$, 
we then ``tensor up'' this matrix to obtain rank lower bounds on $\MC_{tB}$ for $t\in\N$:
To this end, we find the Kronecker power $\MSTAR^{\otimes t}$,
which is a full-rank matrix of dimensions $|\mathcal{I}|^t \times |\mathcal{I}|^t$,
as a submatrix of $\MC_{tB}$.
It follows that $\rank(\MC_n) \geq \Omega(|\mathcal{I}|^{n/B})$ whenever $n$ is divisible by $B$.
For $B=6$, this yields $\rank(\MC_n) \geq \Omega(15^{n/6})=\Omega(1.57^{n})$ over $\Z_p$ when $p \neq 2$.

To proceed, it will be useful to define a particular graph $K_B^{(t)}$ on $tB$ vertices for each $t\in\N$,
which can be viewed as a subgraph of $K_{tB}$.
Only perfect matchings contained in $K_B^{(t)}$ will be relevant.
The graph consists of $t$ disjoint copies of $K_B$
and ``patch edges'' between adjacent $K_B$-copies
that will be used to combine solutions of the individual $K_B$-copies to a global solution.
\begin{defn}
Let $B \in \N$ be fixed. 
For $t\in\mathbb{N}$, let $K_B^{(t)}$ be obtained as follows, see also Figure~\ref{fig:cliqueproduct}:
\begin{enumerate}
\item 
Take $t$ disjoint copies of $K_{B}$ 
and denote the vertices of copy $i\in[t]$ by $v_{i}^{1},\ldots,v_{i}^{B}$.
\item 
For each $i\in[t]$ and $j\in [t]\setminus\{ 1 \}$, 
add an edge from $v_{i}^j$ to $v_{i+1}^1$, 
interpreting $t+1$ as $1$. 
These are the \emph{patch edges}.
\end{enumerate}
\end{defn}
\begin{figure}
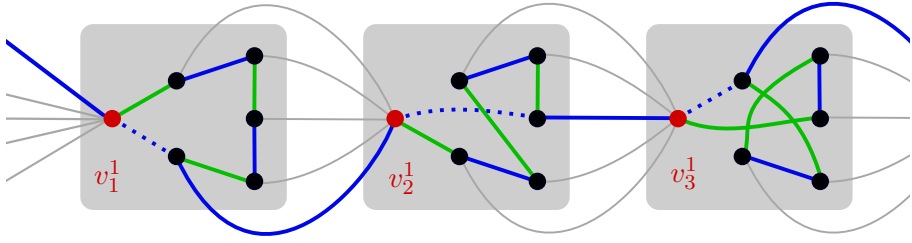

\begin{centering}
\svg{12cm}{product-graph}
\par\end{centering}
\caption{\label{fig:cliqueproduct}The graph $K_B^{(t)}$ for $B=6$ and $t=3$.
Each gray block represents a copy of $K_6$.
Patch edges are drawn with light gray lines; 
the figure wraps around.
Each $K_6$-copy shows a row matching (green) and a column matching (blue).
The special vertex $v_i^{1}$ in the $i$-th copy is drawn red,
and the edge incident with $v_i^{1}$ in the column matching is dotted.
After deleting the dotted edges,
we can choose patch edges, depending only on the column matching, 
so as to obtain a Hamiltonian cycle of $K_B^{(t)}$.}
\end{figure}

The perfect matchings of $K_{B}^{(t)}$ contain a particular subset $\mathcal{I}^{\ominus t}$ 
of size $|\mathcal{I}|^t$
that is essentially the $t$-th power of $\mathcal{I}$;
this set will be the row set of the full-rank submatrix we wish to find in $\MC_{tB}$.
The elements of $\mathcal{I}^{\ominus t}$ are disjoint unions of perfect matchings, 
one for each $K_{B}$-copy.
\begin{defn}
Given a tuple $N = (N_1,\ldots,N_t)\in \mathcal{I}^t$, for $t\in\N$,
we define a perfect matching $M_N^\ominus$ of the graph $K_{B}^{(t)}$ by 
\[M_N^\ominus =\{\{v_{i}^{a},v_{i}^{b}\}\mid i\in[t]\text{ and }\{a,b\}\in N_{i}\}.\]
We write 
$\mathcal{I}^{\ominus t} = \{ M_{N}^\ominus \mid N \in \mathcal{I}^t \}$
for the perfect matchings that can be obtained from $\mathcal{I}^t$ this way.
\end{defn}

It remains to find an appropriate column set of perfect matchings.
Note that we cannot reuse $\mathcal{I}^{\ominus t}$ for this purpose:
If $t\geq 2$, the union of any two perfect matchings in $\mathcal{I}^{\ominus t}$ is disconnected,
and therefore $\MC_{tB}[\mathcal{I}^{\ominus t}, \mathcal{I}^{\ominus t}]$ contains only zeroes.

We do however obtain a suitable column set,
which we denote by $\mathcal{I}^{\obar t}$,
by using the patch edges of $K_B^{(t)}$:
Each perfect matching in $\mathcal{I}^{\obar t}$ is obtained from some $M \in \mathcal{I}^{\ominus t}$
by deleting one particular edge from each $K_B$-copy,
and patching the resulting isolated vertices to adjacent $K_B$-copies.
\begin{defn}
Given a tuple $N = (N_1,\ldots,N_t)\in \mathcal{I}^t$, for $t\in\N$,
we define a perfect matching $M_{N}^\obar$ of the graph $K_{B}^{(t)}$:
\begin{enumerate}
\item Start with the perfect matching $M_{N}^\ominus$.
\item For $i\in[t]$, let $r(i)$ denote the neighbor of $1$ in $N_i$.
Delete the edge from $v_{i}^{1}$ to $v_{i}^{r(i)}$ in $M_{N}^\ominus$,
rendering these two vertices isolated.
\item For $i\in[t]$,
include the patch edge from $v_{i}^{r(i)}$ to $v_{i+1}^{1}$. 
(Consider $t+1=1$ here.)
\end{enumerate}
We then define
$\mathcal{I}^{\obar t} = \{ M_{N}^\obar \mid N \in \mathcal{I}^t \}$.
\end{defn}

It is easily seen that $M_{N}^\obar$ is indeed a perfect matching of the graph $K_{B}^{(t)}$,
for each $N \in \mathcal{I}^t$:
We started with the perfect matching $M_{N}^\ominus$,
then reduced the degree of $v_{i}^{1}$ and $v_{i}^{r(i)}$ to $0$ for all $i \in [t]$, 
and then increased these degrees back to $1$ in the third step.
No other degrees were affected.

Having defined our rows $\mathcal{I}^{\ominus t}$ and columns $\mathcal{I}^{\obar t}$,
we proceed to study the submatrix $\MC_{tB}[\mathcal{I}^{\ominus t},\mathcal{I}^{\obar t}]$.
Note that both $\mathcal{I}^{\ominus t}$ and $\mathcal{I}^{\obar t}$ correspond bijectively to $\mathcal{I}^t$, so the indexing of $\MC_{tB}[\mathcal{I}^{\ominus t},\mathcal{I}^{\obar t}]$ already puts this matrix close to the $t$-th Kronecker power of $\MSTAR$.
Its content also does not fail us:
\begin{lem}
Identifying $\mathcal{I}^{\ominus t}$ and $\mathcal{I}^{\obar t}$ each
with $\mathcal{I}^t$ in the natural way, we have
$\MC_{tB}[\mathcal{I}^{\ominus t},\mathcal{I}^{\obar t}] = \MSTAR^{\otimes t}.$
\end{lem}
\begin{proof}
Given $R,C \in \mathcal{I}^t$ with $R = (R_1, \ldots, R_t)$ and $C = (C_1, \ldots, C_t)$,
let $H = M_R^\ominus \cup M_C^\obar$ be the union of its corresponding perfect matchings.
We observe that $H$ is a Hamiltonian cycle (in $K_{B}^{(t)}$) 
if and only if $R_i \cup C_i$ is a Hamiltonian cycle (in $K_B$) for each $i\in [t]$:

In the ``if'' direction, 
note that $H$ is the result of 
deleting one edge each from $t$ Hamiltonian cycles, 
then adding edges between the endpoints of the resulting Hamiltonian paths 
so as to obtain a Hamiltonian cycle in $K_{B}^{(t)}$.

In the ``only if'' direction, 
note that the restriction of $H$
to the $i$-th $K_B$-copy for $i\in [t]$
is a Hamiltonian path between the $v_i^1$ and some neighbor.
By adding back the edge between $v_i^1$ and its neighbor
and deleting the patch edges,
we obtain a Hamiltonian cycle in each copy of $K_B$.

The claim then follows from the definition of
$\MSTAR = \MC_B[\mathcal{I},\mathcal{I}]$
and the Kronecker product.
\end{proof}

Since $\MSTAR$ has full rank over $\Z_p$
and $p$ was required to be prime,
the Kronecker power $\MSTAR^{\otimes t}$ also has full rank,
so we obtain:

\begin{cor}
The matrix $\MC_{tB}[\mathcal{I}^{\ominus t},\mathcal{I}^{\obar t}]$ has full
rank over $\Z_{p}$.
Consequently, the rank of $\MC_{tB}$ over $\mathbb{Z}_{p}$ is at least $|\mathcal{I}|^t$.
\end{cor}
In conclusion, by using $\MSTAR = \MC_6$,
we obtain that, for prime $p\neq2$, 
the rank of $\MC_{n}$ over $\mathbb{Z}_p$ is at least $\Omega(15^{n/6})=\Omega(1.57^n)$.
Using the larger initial matrices provided by Lemma~\ref{lem:initial-matrices},
we obtain the following stronger bounds:
\begin{thm}
\label{thm:modranks-all}
For prime $p$, the rank of $\MC_{n}$ over $\mathbb{Z}_p$ is at least
\begin{align*}
\Omega(945^{n/10}) & =\Omega(1.984^n) \quad \text{ if } p\notin \{2,3\}, \\
\Omega(3618^{n/12})& =\Omega(1.979^n) \quad \text{ if } p = 3, \\
\Omega(9890^{n/12})& =\Omega(2.152^n) \quad \text{ if } p = 5 ,\\
\Omega(9933^{n/12})& =\Omega(2.153^n) \quad \text{ if } p \in \{7, 11, 13\}.
\end{align*}
\end{thm}
The bounds can be improved by using larger initial matrices $\MSTAR$,
but we hit our computational limit with the $10395 \times 10395$ matrix $\MC_{12}$.
For this matrix,
we could no longer compute determinants of the relevant submatrices to determine their prime factors,
but we could still compute the rank of $\MC_{12}$ for primes up to $13$,
thus obtaining the last three entries in Theorem~\ref{thm:modranks-all}.
\section{The rank of the matchings connectivity matrix over the rational numbers}
\label{sec: reprank}

In this section we establish the first part of Theorem~\ref{thm:mainrank}. For this we need some basics on the representation theory of the symmetric group which we first briefly outline.

\subsection{The Representation Theory of the Symmetric Group}

The representation theory of the symmetric group $S_n$ is remarkable, as much of it may be explained via the combinatorics of \emph{integer partitions} and \emph{tableaux}.  We outline the relevant combinatorial aspects of the theory, leaving the algebraic basics of finite group representation theory for Appendix~\ref{sec:repbackground}. The reader is referred to~\cite{Sagan} for a gentle but more thorough introduction.

Let $\lambda = (\lambda_1,\lambda_2,\cdots,\lambda_k) \vdash n$ such that $\lambda_1 \geq \lambda_2 \geq \cdots \geq \lambda_k$ and $\sum_{i=1}^k \lambda_i = n$ denote an \emph{integer partition} of $n$.  If $j$ parts of the integer partition have the same size $m$, then we express them by the shorthand $m^j$. Let $\lambda(n)$ denote the number of integer partitions of $n$. It is well-known that there is a one-to-one correspondence between the irreducible representations of $S_n$ and the integer partitions of $n$.  We let $[\![  \lambda ]\!]$ denote the irreducible representation of $S_n$ corresponding to $\lambda \vdash n$.

For an integer partition $\lambda \vdash n$, the \emph{Ferrers diagram} of $\lambda$ is an associated left-justified tableau that has $\lambda_i$ cells in the $i$th row.  Abusing notation, we let $\lambda$ also refer the Ferrers diagram of $\lambda \vdash  n$. In Figure~\ref{fig:ferrer} the Ferrers diagram for $(4,3,1^2) \vdash 9$ is illustrated. 

We obtain a \emph{standard Young tableau} from a Ferrers diagram by labeling its cells with numbers such that the numbers along each row are strictly increasing, and the numbers along each column are strictly increasing. In Figure~\ref{fig:young} a standard Young tableau of shape $(4,3,1^2)$ is shown.

Let $f^\lambda$ denote the number of standard Young tableaux of shape $\lambda \vdash n$.  There is an elegant combinatorial formula for expressing $f^\lambda$.  

We say that a tableau $\lambda$ \emph{covers} a tableau $\mu$ if the cells of $\mu$ are contained in the cells of $\lambda$.  A \emph{hook} is a tableau of shape $(k,1^\ell)$, equivalently, a tableau that does not cover the shape $(2^2)$.  The partition $(4,3,1^2)$ is not a hook, as it covers $(2^2)$, illustrated in Figure~\ref{fig:hook}.

For each cell $c \in \lambda$ of a Ferrers diagram, say at row $i$ and column $j$, if we take $c$ along with all cells in row $i$ to the right of $c$, and all cells in column $j$ that lie below $c$, we obtain a hook $(k,1^\ell)$ for some $k,\ell \in \mathbb{N}$. Let $h(c) := k + \ell$ denote the number of cells in this hook, the so-called \emph{hook length}.  In Figure~\ref{fig:hooklength}, we have annotated each cell with its corresponding hook length. 
\begin{figure}
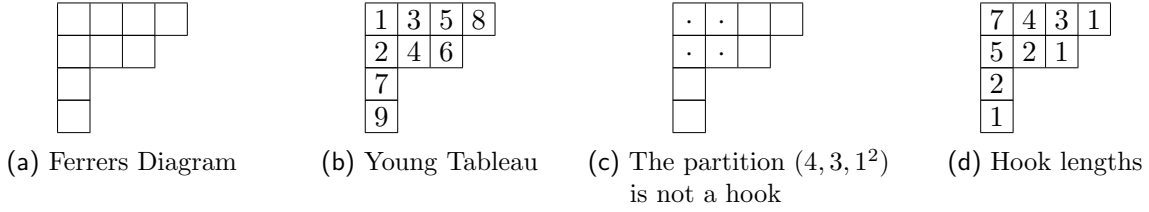

	\centering
	\begin{subfigure}[t]{.24\textwidth}
		\centering
		\young(~~~~,~~~,~,~)
		\caption{Ferrers Diagram}
		\label{fig:ferrer}
	\end{subfigure}
	\begin{subfigure}[t]{.24\textwidth}
		\centering
		\young(1358,246,7,9)
		\caption{Young Tableau}
		\label{fig:young}
	\end{subfigure}
	\begin{subfigure}[t]{.24\textwidth}
		\centering
		\young(\cdot\cdot~~,\cdot\cdot~,~,~)
		\caption{The partition $(4,3,1^2)$ is not a hook}
		\label{fig:hook}
	\end{subfigure}
	\begin{subfigure}[t]{.24\textwidth}
		\centering
		\young(7431,521,2,1)
		\caption{Hook lengths}
		\label{fig:hooklength}
	\end{subfigure}
	\caption{Illustrations of basic notions from the representation theory of the symmetric group}
\end{figure}
The following result connects hook lengths with enumerating standard Young tableaux.
\begin{thm}[Hook Theorem~\cite{Sagan}] $f^\lambda = \frac{n!}{ \prod_{c \in \lambda} h(c)}.$
\end{thm}
\noindent For instance, it is easy to see using the hook formula that $f^\lambda = \binom{n-1}{\ell}$ for any hook $\lambda = (n-\ell,1^\ell) \vdash n$.  A classic result in the representation theory of the symmetric group $S_n$ is that $f^\lambda$ equals the dimension of the irreducible $[\![ \lambda ]\!]$ corresponding to $\lambda$. 
\begin{prop}[Dimensions of Irreducibles of $S_n$~\cite{Sagan}]\label{prop:dim} $f^\lambda = \dim [\![ \lambda ]\!].$
\end{prop}

\subsection{Relating rank to the number of Young tableaux}
We proceed by studying $\MC := \MC_k$ where we let $k=2n$. Let $A, B$ be a partition of the vertices of $K_{2n}$ into two parts of size $n$. Consider the sub-matrix $\MCB$ of $\MC$ induced by the perfect matchings of $K_{2n}$ that are also bipartite perfect matchings with respect to the bipartition $A,B$. In~\cite{RazS95}, Raz and Spieker show that the eigenspaces of $\MCB$ are in fact irreducible representations of $S_n$, and that the eigenspaces corresponding to nonzero eigenvalues of $\MCB$ correspond to the hooks of length $n$. This result paired with some elementary combinatorics implies the following theorem.
\begin{thm}[Raz \& Spieker~\cite{RazS95}]\label{thm:rs}
$$ \rank(\MCB) = \sum_{ \underset{\lambda \emph{ does not cover $(2^2)$}}{\lambda \vdash n}} (f^\lambda)^2 = \binom{2n-2}{n-1}.$$
\end{thm}
\noindent Since there are $\frac{1}{2}\binom{2n}{2}$ ways to partition $V(K_{2n})$ into two parts $A,B$ of size $n$, this already gives an upper bound of $\frac{1}{2}\binom{2n-2}{n-1}\binom{2n}{n}$ on the rank of $\MC$. We will show this is almost tight.\footnote{Moreover, it can be verified experimentally that for some constant $n$ the rank of $\MC$ is strictly smaller than this bound, but lower order terms will not be relevant for us.} One of our key technical theorems is the following exact formula for the rank of $\MC$.
\begin{thm}\label{thm:rank} For any $\lambda \vdash n$, let us write $2\lambda = (2\lambda_1, 2\lambda_2, \cdots, 2\lambda_k) \vdash 2n$. Then 
$$ \rank(\MC) = \sum_{ \underset{\lambda \emph{ does not cover $(2^3)$}}{\lambda \vdash n}} f^{2\lambda}.$$
\end{thm}
\noindent This result can be seen as the non-bipartite analogue of Theorem~\ref{thm:rs}. To prove it, we determine the nonzero eigenvalues of $\MC$; however, this will require a fair amount of algebraic combinatorics, which we now develop.

Let $H_n$ denote the \emph{hyperoctahedral group} of order $n$, equivalently, the group of permutations $\sigma \in S_{2n}$ such that 
$$ \{\{\sigma(1),\sigma(2)\}, \{\sigma (3),\sigma (4)\}, \cdots, \{\sigma (2n-1),\sigma(2n)\}\} = \{\{1,2\}, \{3,4\}, \cdots, \{2n-1,2n\}\}.$$ It is well-known that the set of perfect matchings of $K_{2n}$ can be written as $\mathcal{M}_{2n} := S_{2n} / H_n$.  Even though $S_{2n}/H_n$ is not a group, these cosets possess a remarkable amount of algebraic structure. 

Let $\mathbb{R}[\mathcal{M}_{2n}]$ be the vector space of real-valued functions over perfect matchings, equivalently, the space of real-valued functions over $S_{2n}$ that are \emph{$H_n$-invariant}, that is, $f(\sigma) = f(\sigma h)$ $\forall \sigma \in S_{2n}, h \in H_n$.   In~\cite{Thrall42}, Thrall showed this vector space admits the following decomposition into irreducible representations of $S_{2n}$.
\begin{thm}[Thrall '42]\label{thm:decomp} $\mathbb{R}[\mathcal{M}_{2n}] = \bigoplus_{\lambda \vdash n} [\![2\lambda]\!]$.
\end{thm}
\noindent A consequence of Thrall's result is that $\mathcal{M}_{2n}$ admits a \emph{symmetric association scheme}, the so-called \emph{perfect matching association scheme}~\cite[Section 15.4]{GodsilMeagher}. 
\begin{defn}[Symmetric Association Scheme]
A collection of binary $m \times m$ matrices $A_0, A_1, \cdots, A_d$ is a symmetric association scheme if the following axioms are satisfied.
\begin{enumerate}
\item $A_0 = I$ where $I$ is the identity matrix.
\item $\sum_{i=0}^d A_i = J$ where $J$ is the all-ones matrix.
\item $A_i = A_i^T$ for all $i \in \{1,2,\cdots,d\}$.
\item $A_iA_j$ is a linear combination of $A_0,A_1,\cdots,A_d$ for all $0 \leq i,j \leq d$.
\item $A_iA_j = A_jA_i$ for all $0 \leq i,j \leq d$.
\end{enumerate}
\end{defn}
\noindent We refer the reader to~\cite{BannaiI84, GodsilMeagher} for a more thorough treatment of association schemes. 

Recall that the union of any two perfect matchings is a disjoint union of cycles, which can be represented by an integer partition of the form $2\lambda \vdash 2n$ where $2\lambda = (2\lambda_1, 2\lambda_2, \cdots, 2\lambda_k)$ for some $\lambda \vdash n$.  The perfect matching association scheme is simply the collection of $\mathcal{M}_{2n} \times \mathcal{M}_{2n}$ matrices $\mathcal{A} := \{A_{\lambda}\}_{\lambda \vdash n}$ defined such that $(A_\lambda)_{ij} = 1$ if $i \cup j \cong 2\lambda$, and $0$ otherwise.
\begin{prop}\label{prop:M}
$\MC \cong A_{(n)}$.
\end{prop}
\noindent Since $\mathcal{A}$ is a symmetric association scheme, the eigenspaces of the $A_\lambda$'s coincide, and are precisely the irreducibles in the decomposition given by Thrall~\cite[Section 15.4]{GodsilMeagher}.  In light of this, we can take the distinct eigenvalues of these matrices as column vectors and collect them in a $\lambda(n) \times \lambda(n)$ matrix $P$. For example, when $n = 4$, we have $\mathcal{A} = \{A_{(4)} , A_{(3,1)} , A_{(2,2)} , A_{(2,1^2)} , A_{(1^4)} \}$, and
\[
  P = \kbordermatrix{
    & (4) & (3,1) & (2,2) & (2,1^2) & (1^4) \\
    (4) & 48 & 32 & 12 & 12 & 1 \\
    (3,1) & -8 & 4 & -2 & 5 & 1 \\
    (2,2) & -2 & -8 & 7 & 2 & 1 \\
    (2,1^2) & 4 & -2 & -2 & -1 & 1 \\
    (1^4) & -6 & 8 & 3 & -6 & 1
  }.
\]
For any $\lambda \vdash n$, we call $\Omega_\lambda :=  \{ m \in \mathcal{M}_{2n} : m \cup m^* \cong 2\lambda \}$ the \emph{$\lambda$-sphere}, where 
\[
m^* = \{\{1,2\},\{3,4\},\cdots,\{2n-1,2n\}\}.
\]
The following lemma gives a simple way to determine their size.
\begin{lem}[\cite{LINDZEY2017130}]\label{lem:sphereSize}
Let $l(\lambda)$ denote the number of parts of $\lambda \vdash n$, $m_i$ denote the number of parts of $\lambda$ that equal $i$, and set $z_\lambda :=  \prod_{i\geq1} i^{m_i} m_i!$. Then
\[ \left| \Omega_\lambda \right| = \frac{2^nn!}{2^{l(\lambda)} z_\lambda}.\]
\end{lem}
\noindent Note that the first row of $P$ lists the sizes of the respective spheres.  This is no coincidence, as each $A_\lambda$ has constant row sum $|\Omega_\lambda|$, and so its largest eigenvalue is $|\Omega_\lambda|$ respectively~\cite{GodsilRoyle}. 
It is known that the entries of $P$ are determined by the \emph{zonal spherical functions}~\cite[Chapter VII]{MacDonald95}, which can be thought of as an analogue of irreducible characters in our association scheme setting.
\begin{thm}[\cite{MacDonald95}]\label{thm:eigvals}
Let $\eta_\lambda$ be the eigenvalue of $A_{(n)}$ associated with the $\lambda$-eigenspace.  Then
\[ \eta_\lambda = |\Omega_\lambda| \omega_\lambda^{(n)}.\]
\end{thm}
\noindent A result of Diaconis and Lander is the following explicit formula for determining the value of the zonal spherical function $\omega_\lambda$ evaluated on perfect matchings $m \in \Omega_{(n)}$. 
\begin{lem}[Diaconis \& Lander~\cite{MacDonald95}]\label{lem:diaconis}
For any cell $c \in \lambda \vdash n$, let $w(c)$ denote the number of cells on the same row as $c$ that lie strictly to the left of $c$, and let $n(c)$ denote the number of cells on the same column as $c$ that lie strictly above $c$. Then \[ \omega_\lambda^{(n)} = \frac{1}{2^{n-1}(n-1)!}\sum_{\underset{c \neq (1,1)}{c \in \lambda}} (2w(c) - n(c)),\]
where the sum runs over all cells except for the northwest-most cell $(1,1)$. In particular, we have $\omega_\lambda^{(n)} = 0$ if and only if $\lambda$ covers $(2^3)$.
\end{lem}
\noindent We are now in a position to prove Theorem~\ref{thm:rank}.
\begin{proof}[Proof of Theorem~\ref{thm:rank}]
By Proposition~\ref{prop:M}, we have $\MC = A_{(n)}$. Theorem~\ref{thm:eigvals} and Lemma~\ref{lem:diaconis} together imply that nonzero eigenvalues of $\MC$ do not cover $(2^3)$.  Lemma~\ref{lem:sphereSize} implies that spheres are nonempty, thus these are precisely the nonzero eigenvalues of $\MC$. By Proposition~\ref{prop:dim} and Theorem~\ref{thm:decomp}, the dimension of the eigenspace corresponding to $\eta_\lambda$ is $f^{2\lambda}$, completing the proof.
\end{proof}

\subsection{Counting Young tableaux}
Our combinatorial formula for the rank of the matchings connectivity matrix does not seem to admit any particularly revealing closed-form; nevertheless, we can still get a good lower bound on its rank.    
We say that $2\lambda  \vdash 2n$ such that $\lambda \vdash n$ is a \emph{domino hook} if $\lambda = (k,k,1^{n-2k})$ for some $0 \leq k \leq n$.  For example,
 \[\young(~~~~~~~~,~~~~~~~~,~~,~~,~~,~~,~~)\]
is the Ferrers diagram of the domino hook $2(4,4,1^5)$. Note that if $2\lambda \vdash 2n$ is a domino hook, then $\lambda$ does not cover $(2^3)$. Using WZ-theory~\cite{A=B}, Regev showed that the number of standard Young tableaux of domino hook shape admits an elegant count.
Recall that $C_n = \frac{1}{n+1}\binom{2n}{n}$ is the $n$th Catalan number, and that $\lim_{n \rightarrow \infty} ( (4^n/\sqrt{\pi}n^{3/2}) /C_n) = 1$, see~\cite{StanleyV201}.

\begin{thm}[Regev~\cite{Regev10}]
\[ C_{n-1}C_n = \sum_{ \underset{2\lambda \textnormal{ domino hook}}{2\lambda \vdash 2n}} f^{2\lambda} .\]
\end{thm}

Now we can combine all work and finish the proof:

\begin{proof}[Proof of Theorem~\ref{thm:mainrank}, first part.]
	We have by Theorem~\ref{thm:rank} that
	\[
	 \rank(\MC_n) = \sum_{ \underset{\lambda \text{ does not cover $(2^3)$}}{\lambda \vdash n}} f^{2\lambda} \geq \sum_{ \underset{2\lambda \text{ domino hook}}{2\lambda \vdash 2n}} f^{2\lambda} = C_{n-1}C_n \geq 4^n/n^3,
	\]
	where the inequality follows because if $2\lambda \vdash 2n$ is a domino hook, then $\lambda$ does not cover $(2^3)$.
\end{proof} 
\noindent Note that by the aforementioned upper bound of $\frac{1}{2}\binom{2n-2}{n-1}\binom{2n}{n}$, the lower bound is almost tight.
\section{The reduction}\label{sec:red}
This section is devoted to the proof of Theorem~\ref{thm:genlowhc}. We will show the following stronger lemma:
\begin{lem}\label{modbound}
	Let $p$ be prime and let $r_p \in \R$ be such that $\log_{r_p}(\rank_p(\MC_k))/k \to c$, where $c \geq 1$, as even $k$ tends to infinity.
	Suppose that \CHCP modulo $p$ on graphs with given path decomposition of width $\pw$ can be solved in $O^*((2+r_p -\eps)^{\pw})$ time, for some $\eps>0$. Then there is an $O^*((2-\eps')^n)$ time algorithm that counts the satisfying assignments of a given a CNF-formula on $n$ variables modulo $p$, for some $\eps'>0$ depending on $\eps$.
\end{lem}

Lemma~\ref{modbound} is a generalization of Theorem~\ref{thm:genlowhc} as in SETH we can without loss of generality assume the number of satisfying solutions is at most one as mentioned in Section~\ref{sec:prel}, so for the decision version we can simply check whether the modular count equals $1$ or not. Lemma~\ref{modbound} will be used to prove Theorem~\ref{thm:genlowhc}. It should be noted that many of the gadgets used in the non-innovative parts of this section are heavily based on the lower bound for the decision version from~\cite{CyganKN13}.

\subsection{Illustrated outline of proof.}
Before describing the reduction in detail, we first give an illustrated outline. For this, a basic understanding of previous block propagation schemes as outlined in Section~\ref{sub: lms-reduction} will be advantageous. We start with a high level description of the statement of Lemma~\ref{lem:collem}. The $2^n$ assignments of the variables of the given CNF-formula $\varphi$ are encoded by fingerprints that form a basis in the matrix $\HC_k$. The larger such a basis, the more assignments one can encode for fixed $k$. Lemma~\ref{lem:collem} asserts the existence of a certain graph $G$ in which the number of partial solutions of a given fingerprint equals $0$ if the fingerprint encodes an assignment not satisfying $\varphi$ and a fixed positive quantity (depending on the fingerprint) otherwise. The boundary vertices $L$ and $R$ are partitioned into blocks and the fingerprints will also have some block-structure for technical reasons reminiscent to the block-propagation scheme. Confer the statement of Lemma~\ref{lem:collem} for the precise details.

\begin{figure}
	\centering
	\begin{subfigure}[t]{.45\textwidth}
		\centering
		\begin{tikzpicture}[scale=0.5]
		\tikzstyle{bn}=[circle,fill=gray,text=white,draw=black]
\tikzstyle{wn}=[circle,fill=white,draw]
\tikzstyle{big2}=[circle,fill=gray,minimum size=1.2cm, draw=black,text=white]
\tikzstyle{mythick}=[line width=4]
\tikzstyle{mydashed}=[dashed, line width=4]
\tikzstyle{mydotted}=[dotted, line width=1]

\foreach \x in {1.95,7}
	\draw[mydotted] (\x,-8.25) -- (\x,-9.25);

\foreach \layer/\offset in {1/0, 2/7.5, q/18}
{
\begin{scope}[yshift=-\offset cm]
\coordinate (k\layer) at (1.95, 3.5);
\end{scope}
}

\draw[rounded corners,color=gray] (3+0.1,7.5) rectangle (3 + 3-0.1, -17.75);
\draw (3 + 1.5, -5) node {$G$};

\foreach \layer/\offset in {1/0, 2/7.5, q/18}
{
\begin{scope}[yshift=-\offset cm]

\draw[rounded corners,color=gray] (1,0.25) rectangle (3-0.1, 7.5);
\draw[rounded corners,color=gray] (6.1,0.25) rectangle (8, 7.5);

\foreach \y in {1, ..., 6}
	\node[wn,inner sep=0pt, minimum size=0.3cm] (l\layer\y) at (3, 7.8-\y*1.1) {\tiny $l_{\layer,\y}$};

\foreach \y in {1, ..., 6}
	\node[wn,inner sep=0pt, minimum size=0.3cm] (r\layer\y) at (6, 7.8-\y*1.1) {\tiny $r_{\layer,\y}$};
\end{scope}
}

\foreach \i/\offset/\t in {1/0/1, 2/-7.5/2, 3/-10.5/{q-1}} 
{
	\node[wn,inner sep=0pt, minimum size=0.7cm] (y\i) at (2,\offset+0.25) {\tiny $b_{\t}$};
}

\node[wn] (yq) at (-1,-3.5) {};
\draw(yq) edge[bend left=40] (l11);
\draw[dashed] (yq) edge[bend right=40] (lq3);
\draw ($(yq)-(1,0)$) node {\tiny $b_q$};

\draw[dashed] (y1) edge[bend left=40] (l13);
\draw (y1) edge[bend right=40] (l21.220);
\draw (y2)[dashed] edge[bend left=40] (l23);

\draw (y3) edge[bend right=40] (lq1.220);

\draw (8.7, 6) node {$G^1_R$};
\draw (8.7, -3.5) node {$G^2_R$};
\draw (8.7, -16) node {$G^q_R$};
\draw (0.45, 6) node {$G^1_L$};
\draw (0.45, -3.5) node {$G^2_L$};
\draw (0.45, -16) node {$G^q_L$};

\draw (l15.180) -- (1.5,2);
\draw (l15.180) -- (1.5,2.5);

\draw (r15.180) -- (5,2);
\draw (r15.180) -- (5,2.5);

\draw (l16.0) -- (4,1);
\draw (l16.0) -- (4,1.35);

\draw (r16.0) -- (7,1);
\draw (r16.0) -- (7,1.35);

\draw (l12.180) edge[bend right=40]  (l14.180);
\draw (r12.180) edge[bend right=40]  (r14.180);

\draw (l11) -- (r11);

\draw (l12) -- (r13);

\draw[dashed] (r11.0) edge[bend left=60] (r12.0);
\draw (l13.0) edge[bend left=60] (l14.0);
\draw (r13.0) edge[bend left=60] (r14.0);

\draw (r21.0) edge[bend left=60] (r22.0);
\draw (rq1.0) edge[bend left=60] (rq2.0);
		\end{tikzpicture}
		\caption{From Lemma~\ref{lem:collem} to Theorem~\ref{thm:genlowhc}}
		\label{fig:collemtotheorem2}
	\end{subfigure}
	\begin{subfigure}[t]{.3\textwidth}
		\centering
		\begin{tikzpicture}[scale=0.5]
		\tikzstyle{bn}=[circle,fill=gray,text=white,draw=black]
\tikzstyle{wn}=[circle,fill=white,draw]
\tikzstyle{big2}=[circle,fill=gray,minimum size=1.2cm, draw=black,text=white]
\tikzstyle{mythick}=[line width=4]
\tikzstyle{mydashed}=[dashed, line width=4]
\tikzstyle{mydotted}=[dotted, line width=2]

\foreach \layer/\offset in {1/0}
{
\begin{scope}[yshift=-\offset cm]
\coordinate (k\layer) at (1.95, 3.5);
\end{scope}
}

\draw[rounded corners,color=gray] (3+0.1,7.5) rectangle (3 + 3-0.1, -1.75);
\draw (3 + 1.5, 1) node {$G^1_L$};

\draw[rounded corners,color=gray] (6+0.1,7.5) rectangle (6 + 3-0.1, -1.76);
\draw (6 + 1.5, 1) node {$G_R^2$};

\foreach \layer/\offset in {1/0}
{
\begin{scope}[yshift=-\offset cm]

\foreach \y in {1, ..., 8}
	\node[wn,inner sep=0pt, minimum size=0.3cm] (l\layer\y) at (3, 7.8-\y*1.1) {\tiny $l_{\layer,\y}$};

\foreach \y in {1, ..., 8}
	\node[wn,inner sep=0pt, minimum size=0.3cm] (r\layer\y) at (6, 7.8-\y*1.1) {\tiny $r_{\layer,\y}$};
	
\foreach \y in {1, ..., 8}
	\node[wn,inner sep=0pt, minimum size=0.3cm] (s\layer\y) at (9, 7.8-\y*1.1) {\tiny $s_{\layer,\y}$};
\end{scope}
}

\draw (l15.0) edge[bend left=40]  (l16.0);

\draw (r12.180) edge[bend right=40]  (r15.180);
\draw (r14.180) edge[bend right=40]  (r16.180);
\draw (l15.0) edge[bend left=40]  (l16.0);

\draw[decoration = {zigzag,segment length = 1mm, amplitude = 0.75mm}, decorate] (l11) -- (r11);

\draw (l12) -- (r13);

\draw (l13.0) edge[bend left=60] (l14.0);
\draw (r13.0) edge[bend left=60] (r16.0);
\draw (r14.0) edge[bend left=60] (r15.0);

\draw[decoration = {zigzag,segment length = 1mm, amplitude = 0.75mm}, decorate] (r11) -- (s11);
\draw (r12) -- (s13);

\draw (s12.180) edge[bend right=40]  (s14.180);
\draw (s15.180) edge[bend right=40]  (s16.180);
\draw (s15.0) edge[bend left=40]  (s16.0);

\draw (l17.180) -- (2,0.3);
\draw (l17.180) -- (2,-0.1);

\draw (r17.180) -- (5,0.3);
\draw (r17.180) -- (5,-0.1);

\draw (s17.180) -- (8,0.3);
\draw (s17.180) -- (8,-0.1);

\draw (l18.0) -- (4,-0.8);
\draw (l18.0) -- (4,-1.2);

\draw (r18.0) -- (7,-0.8);
\draw (r18.0) -- (7,-1.2);

\draw (s18.0) -- (10,-0.8);
\draw (s18.0) -- (10,-1.2);
		\end{tikzpicture}
		\caption{Lemma~\ref{lem:collem}, inductive step}
		\label{fig:colleminductive}
	\end{subfigure}
	\begin{subfigure}[t]{.23\textwidth}
		\centering
		\begin{tikzpicture}[scale=0.5]
		\input{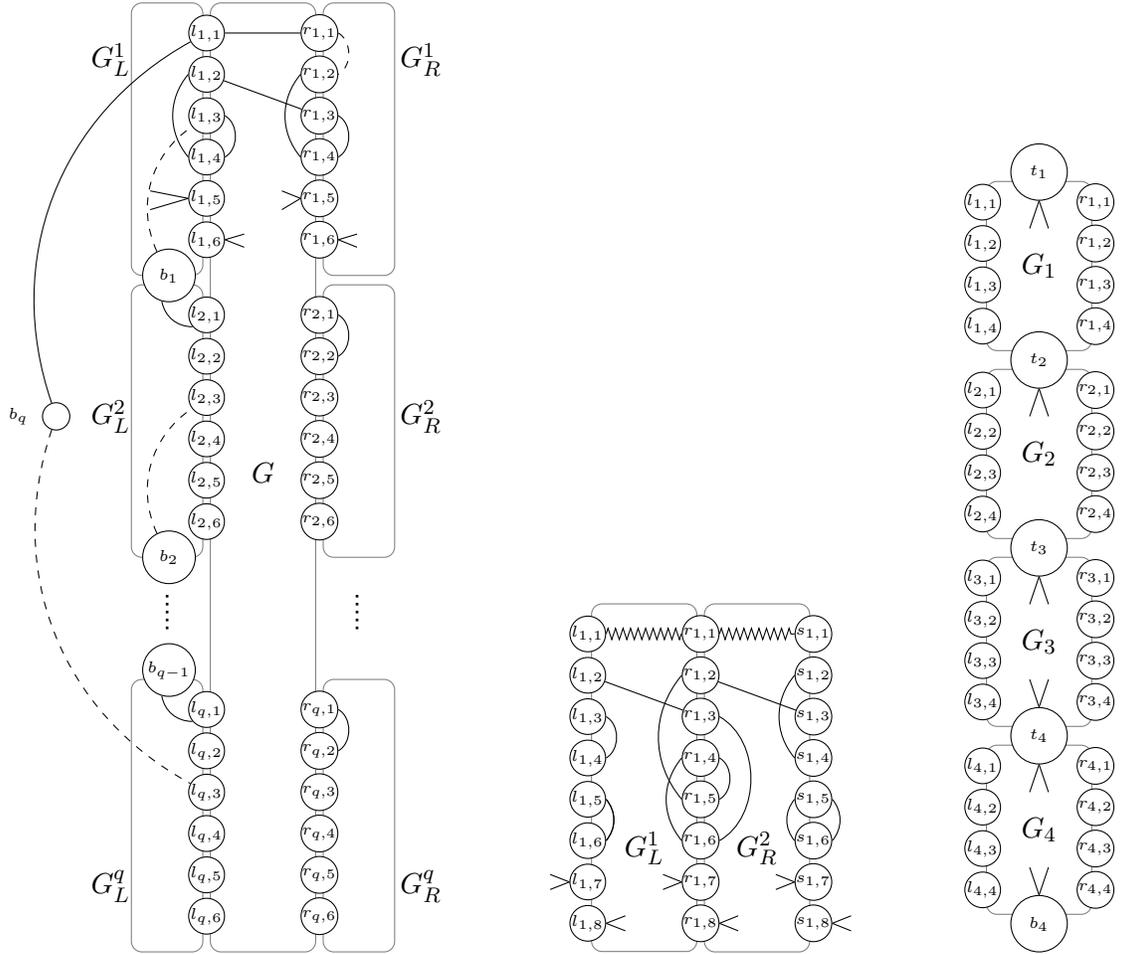}
		\end{tikzpicture}
		\caption{Lemma~\ref{lem:collem}, base case}
		\label{fig:collembase}
	\end{subfigure}
	\caption{Illustrations of parts of the proof of Theorem~\ref{thm:genlowhc}.}
\end{figure}

In Figure~\ref{fig:collemtotheorem2}, we illustrate how the graph output by the reduction implied by Lemma~\ref{modbound} is obtained from the graph $G$ obtained by Lemma~\ref{lem:collem}: The \emph{blocks} $G^i_L$ and $G^i_R$ (whose unions we jointly denote by $G_L$ respectively $G_R$) are added to $G$, where $G^i_R$ has vertex set $R_i$ common with $G$, and $G^i_L$ has vertex set $L_i$ common with $G$. Additionally, graphs $G^i_L$ and $G^{i+1}_L$ share a common vertex $b_{i}$.
By Lemma~\ref{lem:collem}, the graph $G$ has $\mathbf{A}[f_L,f_R]$ partial solutions for fingerprints $f_L$ and $f_R$. 
The graph $G_L$ has $\mathbf{l}_f$ fingerprints and $\mathbf{r}_{f'}$ fingerprints such that $\mathbf{l}^T\mathbf{A}\mathbf{r}$ equals the number of satisfying solutions of $\varphi$ modulo $p$. To establish this, for each fingerprint, certain partial solutions are allowed such that two edgesets avoid creating subcycles only if they have combining fingerprints on $L$ (or $R$). The structure of the partial solutions is tailored so that combinations of partial solutions with combining fingerprints can be extended in one (modulo $p$) way into a Hamiltonian cycle. For the latter part, the vertices $b_i$ are used to connect parts from the subgraphs $G^i_L$. In Figure~\ref{fig:collemtotheorem2}, the edges present in all partial solutions are displayed everywhere (where dashed edges will actually be redirected to visit vertices with degree $0$ or $2$ in a fingerprint), and on the vertices $L_i \cup R_i$, a possible partial solution in $G^i_L,G,G^i_R$ is displayed. The edges $l_{i,1}l_{i,2}$ and $r_{i,1}r_{i,3}$ in $G$ are rerouted to avoid creating two subcycles if two matchings give exactly one cycle. If the fingerprint in both $G^1_L$ and $G^1_R$ matches with that in $G$, $l_{1,1}$ will be connected to $r_{1,1}$ and $l_{1,2}$ will be connected to $r_{1,2}$.

In Figure~\ref{fig:colleminductive}, we illustrate the inductive step of the proof of Lemma~\ref{lem:collem}.
This lemma states that, for a CNF-formula $\varphi=C_1\wedge\ldots\wedge C_m$ on $m$ clauses, there exists a graph such that the number of partial solutions for certain fingerprints encodes the number of satisfying assignments of $\varphi$.
The proof of this lemma is by induction on $m$, so in the inductive step, we assume a graph $G_L$ is given for $\varphi'=C_1\wedge\ldots\wedge C_{m-1}$, and a graph $G_R$ for the CNF-formula $C_m$.
The graphs $G_L$ and $G_R$ have a certain block structure consisting of smaller graphs $G^i_L$ and $G^i_R$, which allows us to basically restrict attention to one block (i.e., a fingerprint of one block encodes a partial assignment of a block of variables).\footnote{This is not completely true, as partial assignments must interact to check whether any partial assignment satisfies a clause, but this does not turn out problematic.} 
Suppose the graphs $G^i_L$ and $G^i_R$ have boundary vertices $L \cup R$ and $R \cup S$ respectively. Then we can see these graphs as matrices $\mathbf{L}$ and $\mathbf{R}$ of which an entry indexed by $f,f'$ equals the number of partial solutions on fingerprints $f,f'$ in $G^i_L$ (in which case $f$ is on $L$ and $f'$ is on $R$), and respectively the number of partial solutions on fingerprints $f,f'$ in $G^i_R$ (in which case $f$ is on $R$ and $f'$ is on $S$). Note the fingerprint of the partial solution in $G^i_L$ on $R$ and the fingerprint of $G^i_R$ on $R$ also need to match as otherwise the union cannot give a partial solution in the graph $G^i \cup G^i_R$, and by our construction two partial solutions with matching fingerprint on $R$ are also a partial solution of $G^i \cup G^i_R$. It follows that if we let $\mathbf{A}$ be the number of partial solutions in $G^i \cup G^i_R$ indexed by fingerprint $f_L$ on $L$ and $f_R$ on $R$, than this matrix can be computed as the matrix multiplication $\mathbf{L}\FR\mathbf{R}$, where $\mathbf{F}$ is a full rank submatrix of the Hamiltonian fingerprint matrix by a set of fingerprints that we restrict ourselves to. By using gadgets to ensure that $\mathbf{L}=\mathbf{C}\FR^{-1}$ where $\mathbf{C}$ is a diagonal matrix checking whether the assignment satisfies $\varphi'$ the inductive step can be carried out. In Figure~\ref{fig:colleminductive} the showed partial solution in $G^1_L$ is derived from the fingerprint $f_L=(d_L,M_L)$ in $G^1_L$ with $d^{-1}_L(1)=\{l_{1,i}\}^6_{i=1}$, $d^{-1}_L(2)=\{l_{1,8}\}$, $d^{-1}_L(0)=\{l_{1,7}\}$ and $M_L=\{l_{1,1}l_{1,2},l_{1,3}l_{1,4},l_{1,5}l_{1,6}\}$ on $L$ and from the fingerprint $f_R=(d_R,M_R)$ in $G^1_L$ with $d_R(v)=2-d_L(v)$, and $M_R=\{r_{1,1}r_{1,3},r_{1,2}r_{1,5},r_{1,4}l_{1,6}\}$ on $L$. If this (or a fingerprint of another block) encodes a partial assignment satisfying $\varphi'$, there will be $\FR^{-1}[f_L,f_R]$ partial solutions with fingerprint $f_L \cup f_R$ in $G^1_L$.

In Figure~\ref{fig:collembase} we illustrate the base case of the proof of Lemma~\ref{lem:collem}. The graph is partitioned into $q$ blocks $G_1,\ldots,G_q$ (with $q=4$) in the example, and a fingerprint on $L_i=\{l_{i,j}\}_\{j\leq \beta \}$ encodes a partial assignment of a block of variables of the CNF-formula, which we consider to be a single clause as $m=1$ in the base case of the induction. A graphs $G_i$ contains a top vertex $t_i$ and bottom vertices $b_i$, and the consecutive graphs overlap in the sense that $b_{i+1}=t_i$. Partial solutions are locally (that is, per $G_i$) restricted such that $i$ is the smallest integer with the property that the partial solution encodes a fingerprint satisfying the clause if and only if the partial solution in $G_i$ has a fingerprint in which both $t_i$ and $b_i$ are of degree $2$.

\subsection{Pattern propagation using a rank lower bound}\label{subsec:patternpropagation}
\newcommand{\lss}{l}
\newcommand{\rss}{r}
Let $\gamma:=\gamma(\eps) \geq 4$ be even. Assume we are given a CNF-formula $\varphi=C_1 \wedge \ldots \wedge C_m$ on variables $x_1,\ldots,x_n$ with $n$ being a multiple of $\gamma$;\footnote{Note that since $\gamma$ is a constant this is easily established by adding at most $\gamma$ dummy variables.} let~$q=\frac{n}{\gamma}$. Partition the set $\{x_1,\ldots,x_n\}$ into $n/\gamma$ blocks of size $\gamma$, denoted $X_1,\ldots,X_{n / \gamma}$. Intuitively, we will represent the $2^\gamma$ assignments of the block of variables $X_i$ by HC-fingerprints on groups of vertices in a bag of the to be constructed path decomposition.
We let $\cB_\lss,\cB_\rss$ be sets of HC-fingerprints on $[\beta]$ such that $\HC[\cB_\lss,\cB_\rss]$ has full rank over $p$ and if $(d_\lss,M_\rss) \in \cB_\lss$ and $(d_\rss,M_\rss) \in \cB_\rss$ then
\begin{itemize}
	\item $d_\lss(i), d_\rss(i)$ equal $1$ for $i=1,2,3$, and
	\item $\{1,3\} \in M_\lss$ and $\{1,2\} \in M_\rss$.
\end{itemize}
For $i=1,\ldots,q$ we assume $\eta_i$ is an injective function from $\cB_\rss$ to $\{0,1\}^{X_i}$ which describes the encoded partial assignment to the variable set $X_i$. Note this is only possible when $|\cB_\rss| \geq 2^{\gamma}$, which we will ensure in Subsection~\ref{sec:puttog} when wrapping up the proof of Theorem~\ref{thm:genlowhc}.
We let denote $\eta$ denote the natural extension of $\eta_i$, i.e. $\eta(f_1,\ldots,f_{n/\gamma})=(\eta_1(f_1),\ldots,\eta_{n/\gamma}(f_{n/\gamma}))$.
 In the following, we refer to $\HC[\cB_\lss,\cB_\rss]$ as $\FR$, and also frequently use the $q$-th Kronecker power $\FR^{\otimes q}$. We also use $\FR^{-1}$ (which is indexed by $\cB_\rss$ and $\cB_\lss$ respectively) and its Kronecker power $(\FR^{-1})^{\otimes q} = (\FR^{\otimes q})^{-1}$.

\begin{lem}\label{lem:collem}
	Let $L=L_1 \cupdot  \ldots \cupdot L_q, R=R_1 \cupdot  \ldots \cupdot R_q$ be disjoint sets with $L_i=\{l_{i,1},\ldots,l_{i,\beta}\}$ and $R_i=\{r_{i,1},\ldots,r_{i,\beta}\}$.
	There exists a graph $G(L,R,\varphi)=(V,E)$ with independent set $L \cup R \subseteq V$ such that for every sequence of fingerprints 
	\[
		f_L = \left((d^i_L,M^i_R)) \text{ on $L_i$}\right)_{i \in [q]} \text{ and }  f_R = \left((d^i_R,M^i_R) \text{ on $R_i$}\right)_{i \in [q]},
	\]
	the number of partial solutions in $G$ for fingerprint $f' = \left(\bigcup_{i=1}^q d^i_L \cup d^i_R, \cup_{i=1}^q M_{alt}^i\right)$ equals
	\[ \mathbf{A}[f_L,f_R] :=
		\begin{cases}
			(\FR^{\otimes q})^{-1}[f_L,f_R], & \text{ if } f_L \in \cB_\rss^{q}, f_R \in \cB_\lss^{q}  \text{ and } \eta(f_L) \text{ satisfies } \varphi,\\
			0, & \text{otherwise}.
		\end{cases}
	\]
	Here $M_{alt}^i$ is the altered matching defined by $M_{alt}^i = M^i_L \cup M^i_R \setminus \{l_{i,1}l_{i,2},r_{i,1}r_{i,3}\} \cup \{l_{i,1}r_{i,1},l_{i,2}r_{i,3}\}$. Moreover, $G$ has a path decomposition $\mathbb{P}$ of width $\beta\cdot q+O(\beta)$ starting in $L$ and ending in $R$, and $G$ and $\mathbb{P}$ can be computed in polynomial time.
\end{lem}

It may look counterintuitive (or like a typo) that $f_L \in \cB_\rss^{q}$, but note that the subscripts $l$ and $r$ denote that fingerprints in $\cB_l$ will be used for partial solutions connecting vertices `to the left' of a certain vertex boundary, while fingerprints in $\cB_r$ will be used for partial solution `to the right' of a certain boundary (see Figure~\ref{fig:colleminductive}). As $L$ is the `left boundary' of $G$ connection made in $G$ between vertices in $L$ will be `to the right' of $L$ so the matchings will be in $\cB_r$.

We will prove Lemma~\ref{lem:collem} by induction on the number of clauses $m$. The bulk of our technical efforts with gadgets will be to prove the following lemma for $m=1$. On the other hand, for proving the inductive step our new key insight of applying matrix inversion is crucial. We postpone the technical proof of the base case $m=1$ to Subsection~\ref{sec:lemcollembase} and first focus on the inductive step.

\begin{proof}[Proof of Lemma~\ref{lem:collem} (inductive step)]
	Let $\varphi = C_1 \wedge \ldots \wedge C_m$ be a CNF formula on $n$ variables and $\varphi'=C_1 \wedge \ldots \wedge C_{m-1}$.
	Let $L=L_1 \cupdot  \ldots \cupdot L_q, R=R_1 \cupdot  \ldots \cupdot R_q,S=S_1 \cupdot  \ldots \cupdot S_q$ be disjoint sets with $L_i=\{l_{i,1},\ldots,l_{i,\beta}\}$, $R_i=\{r_{i,1},\ldots,r_{i,\beta}\}$ and $S_i=\{s_{i,1},\ldots,s_{i,\beta}\}$. Let $\hat{G}=(\hat{V},\hat{E})=G(L,R,\varphi')$ and $\tilde{G}=(\tilde{V},\tilde{E})=G(R,S,C_m)$ be graphs as given by the induction hypothesis.
	
    Define $G=(\hat{V} \cup \tilde{V},\hat{E} \cup \tilde{E})$. 
    We show $G$ satisfies the conditions of Lemma~\ref{lem:collem} with $\varphi$, as required to prove the inductive step.
    Note that $\hat{E} \cap \tilde{E} = \emptyset$ as $\hat{V} \cap \tilde{V} = R$ is an independent set in both graphs. 
    As $\hat{V} \cap \tilde{V}=R$, the path decomposition of $\hat{G}$ ending in $R$ and the path decomposition of $\tilde{G}$ starting in $R$ can clearly be combined into a path decomposition of $G$ of the same width.
    It remains to show that for every sequence of fingerprints
    \[
    	f_L = \left((d^i_L,M^i_L) \text{ on $L_i$}\right)_{i \in [q]} \text{ and }  f_S = \left((d^i_S,M^i_S) \text{ on $S_i$}\right)_{i \in [q]},
    \]
    the number of partial solutions $H$ in $G$ for fingerprint $f = \left(\bigcup_{i=1}^q d^i_L \cup d^i_S, \cup_{i=1}^q M_{alt}^i\right)$ equals
    \begin{equation}\label{eq:toshow}
    \begin{cases}
    (\FR^{\otimes q})^{-1}[f_L,f_S], & \text{ if } f_L \in \cB_\rss^{q}, f_S \in \cB_\lss^{q}  \text{ and } \eta(f_L) \text{ satisfies } \varphi,\\
    0, & \text{otherwise}.
    \end{cases}
    \end{equation}
    Here $M_{alt}^i = M^i_L \cup M^i_S \setminus \{l_{i,1}l_{i,2},s_{i,1}s_{i,3}\} \cup \{l_{i,1}s_{i,1},l_{i,2}s_{i,3}\}$.
    To show this, we first show the following:
    \begin{claim}\label{clm:psnr} The number of partial solutions in $G$ for fingerprint $f$ equals
    	\[
    		\sum_{\substack{ \hat{f}_{R}\in \cB^{\otimes q}_\lss,\tilde{f}_{R} \in \mathcal{B}^{\otimes q}_\rss \\ \eta(f_L) \models \varphi',\  \eta(\tilde{f}_R) \models C_m}} (\FR^{\otimes q})^{-1}[f_L,\hat{f}_R] \cdot \FR^{\otimes q}[\hat{f}_R,\tilde{f}_R] \cdot (\FR^{\otimes q})^{-1}[\tilde{f}_R,f_S].
    	\]
    \end{claim}
    \begin{proof}    
    Let $H$ be a partial solution in $G$ for fingerprint $f$, let $\hat{H} = H \cap \hat{E}$ and $\tilde{H} = H \cap \tilde{E}$.
    By construction of $\hat{G}$ and $\tilde{G}$, there are fingerprints $\hat{f}$ and $\tilde{f}$ of the form
    \[
    	\hat{f} = \left(\bigcup_{i=1}^q d^i_L \cup \hat{d}^i_R, \cup_{i=1}^q \hat{M}^i_{alt}\right), \quad \tilde{f} = \left(\bigcup_{i=1}^q \tilde{d}^i_{R} \cup d^i_S, \cup_{i=1}^q \tilde{M}_{alt}^i\right),
    \]
    such that $\hat{H}$ is a partial solution in $\hat{G}$ for $\hat{f}$ and $\tilde{H}$ is a partial solution in $\tilde{G}$ for $\tilde{f}$. Here the altered matchings are of the form
    \begin{align*}
    	\hat{M}^i_{alt} &= M^i_L \cup \hat{M}^i_R \setminus \{l_{i,1}l_{i,2},r_{i,1}r_{i,3}\} \cup \{l_{i,1}r_{i,1},l_{i,2}r_{i,3}\},\\
    	\tilde{M}^i_{alt} &= M^i_{R} \cup M^i_S \setminus \{r_{i,1}r_{i,2},s_{i,1}s_{i,3}\} \cup \{r_{i,1}s_{i,1},r_{i,2}s_{i,3}\}.
    \end{align*}
    As $H$ is a partial solution in $G$ we have $d_H(v)=2$ for all vertices in $v \in R$ and thus $\hat{d}^i_R(j)+\tilde{d}^i_{R}(j)=2$ for every $i,j$. Moreover, $H$ cannot contain a cycle as it is a partial solution and therefore $\hat{M}^i_{alt} \cup \tilde{M}^i_{alt}$ cannot contain a cycle for every $i=1,\ldots,q$. It follows that $\hat{M}^i_R \cup \tilde{M}^i_{R}$ must form a single cycle: if not, it contains at least two cycles as $\hat{M}^i_R \cup \tilde{M}^i_{R}$ are perfect matchings on the same set of vertices, and a cycle not containing the vertex $r_{i,1}$ will still be present in $\hat{M}^i_{alt} \cup \tilde{M}^i_{alt}$. Thus in summary we have that if $\hat{H}$ is a partial solution in $\hat{G}$ for fingerprint $\hat{f}_R=(\cup_{i=1}^q \hat{d}^i_R,\cup_{i=1}^q \hat{M}^i_R)$ on $R$ and is $\tilde{H}$ is a partial solution in $\tilde{G}$ for fingerprint $\tilde{f}_R=(\cup_{i=1}^q \tilde{d}^i_R,\cup_{i=1}^q \tilde{M}^i_R)$ on $R$, then $\FR^{\otimes q}[\hat{f}_R,\tilde{f}_R]=1$.
    
    For the reverse direction we have that if $\FR^{\otimes q}[\hat{f}_R,\tilde{f}_R]=1$, by the definition of the altered matchings $H$ indeed has fingerprint $f=\left(\bigcup_{i=1}^q d^i_L \cup d^i_S, \cup_{i=1}^q M_{alt}^i\right)$ where 
    \[ 
    	M_{alt}^i = M^i_L \cup M^i_S \setminus \{l_{i,1}l_{i,2},s_{i,1}s_{i,3}\} \cup \{l_{i,1}s_{i,1},l_{i,2}s_{i,3}\}.
    \]
    To see this, note that in $H$, the vertex $l_{i,1}$ is connected to $r_{i,1}$ (directly via $\hat{M}_{alt}^i$), which is connected to $s_{i,1}$ (directly via $\tilde{M}_{alt}^i$), and $l_{i,2}$ is connected to $r_{i,3}$ (directly via $\hat{M}_{alt}^i$), which is connected to $r_{i,2}$ (indirectly via the path $(\hat{M}_{alt}^i \cup \tilde{M}_{alt}^i) \cap R \times R$), which is connected to $s_{i,3}$ (directly via $\tilde{M}_{alt}^i$).
    
    The claim follows by summing over all $\hat{f}_{R}$ and $\tilde{f}_{R}$ such that $\eta(f_L) \models \varphi'$ and $\eta(\tilde{f}_R) \models C_m$ (the latter two properties follow by the properties of $\hat{G}$ and $\tilde{G}$).    
	\end{proof}
    By Claim~\ref{clm:psnr} the number of partial solutions in $G$ for fingerprint $f$ equals $\mathbf{A}[f_L,f_S]$, where
    \[
    	\mathbf{A} = \mathbf{C} (\FR^{\otimes q})^{-1} \FR^{\otimes q} \mathbf{C'} (\FR^{\otimes q})^{-1} = \mathbf{C}\mathbf{C'}(\FR^{\otimes q})^{-1},
    \]
	and $\mathbf{C},\mathbf{C'}$ are diagonal matrices defined by
	\[
		\mathbf{C}[f,f']=\begin{cases}
			1, & \text{if } f=f' \text{ and } \eta(f) \models \varphi' \\
			0, & \text{otherwise }
		\end{cases}
	\]
	and
	\[
	\mathbf{C'}[f,f'] = 
	\begin{cases}
	1, & \text{if } f=f' \text{ and } \eta(f') \models C_m \\
	0, & \text{otherwise. }
	\end{cases}
	\]
	and the inductive step follows as $\mathbf{A}[f_L,f_S]$ equals~\eqref{eq:toshow}.
\end{proof}

\subsection{Gadgets}
\label{subsec:gadgets}
\newcommand{\tsq}{\ {\tiny $\square$}}
\newcommand{\tbsq}{{\tiny $\blacksquare$}}

In this section we introduce two general gadgets, adapted from previous constructions~\cite{CyganKN13}, that are used in the final construction to obtain strong control on the number of Hamiltonian cycles. Both gadgets accept parameters to be set in the final construction.

\paragraph{Label Gadget.}
The following gadget allows us to label incident edges of a vertex $v$ and control label combinations of the edges used in a Hamiltonian cycle. 

\begin{defn}
	A \emph{label gadget} is a pair $(v,\lambda_v)$ where $\lambda_v: I(v) \rightarrow \{1,2,3,4\}$ is a labeling of the edges $I(v)$ incident with $v$. A Hamiltonian cycle $C$ is \emph{consistent} with label
	gadget $(v,\lambda_v)$ if $\lambda_v(e)=2i$ and $\lambda_v(e')=2i-1$ for $i \in \{1,2\}$, where $e,e'$ are the two edges of $C$ incident with $v$.
\end{defn}

\begin{figure}
	\begin{center}
		\begin{tikzpicture}[scale=0.45]

\tikzstyle{bn}=[circle,fill=gray,text=white,draw=black]
\tikzstyle{wn}=[circle,fill=white,draw]
\tikzstyle{big}=[circle,minimum size=5cm, draw]
\tikzstyle{mythick}=[line width=4]
\tikzstyle{mydashed}=[dashed, line width=4]
\tikzstyle{lab1}=[{Square[scale=1.2]}-]
\tikzstyle{lab2}=[{Square[scale=1.2,open]}-]
\tikzstyle{lab3}=[{Circle[scale=1.2]}-]
\tikzstyle{lab4}=[{Circle[scale=1.2,open]}-]

\foreach \y/\t in {1/1, 2/2, 3/3}
{
	\node[wn] (v\t) at (0, 2*\y) {$v_\t$};
}

\node[wn] (v1) at (0, 2) {$v_1$};
\node[wn] (v2) at (8, 2) {$v_2$};
\node[wn] (v3) at (0, 6) {$v_3$};
\node[wn] (v4) at (8, 6) {$v_4$};
\node[wn] (v5) at (0, 4) {$v_5$};
\node[wn] (v6) at (2, 4) {$v_6$};
\node[wn] (v7) at (4, 4) {$v_7$};
\node[wn] (v8) at (6, 4) {$v_8$};
\node[wn] (v9) at (8, 4) {$v_9$};

\draw (v1) -- (v5) -- (v3);
\draw (v6) -- (v7) -- (v8);
\draw (v4) -- (v9) -- (v2) -- (v8);
\draw (v1) -- (v6);
\draw (v3) edge[bend left=20] (v8);
\draw (v6) edge[bend left=20] (v4);

\foreach \d in {-0.1,0.1} {
	\draw[lab1] (v1) -- ($(v1)+(3*\d,-1.5)$);
	\node () at ($(v1)+(4*\d,-1.875)$) {$1$};
	\draw[lab2] (v2) -- ($(v2)+(3*\d,-1.5)$);
	\node () at ($(v2)+(4*\d,-1.875)$) {$2$};
	\draw[lab3] (v3) -- ($(v3)+(3*\d,1.5)$);
	\node () at ($(v3)+(4*\d,1.875)$) {$3$};
	\draw[lab4] (v4) -- ($(v4)+(3*\d,1.5)$);
	\node () at ($(v4)+(4*\d,1.875)$) {$4$};
}

\end{tikzpicture}
	\end{center}
	\caption{Implementation of the label gadget. The four labels $\{1,2,3,4\}$ are depicted symbolically with \tbsq\, \tsq, $\bullet$, and $\circ$ respectively for reference in Figure~\ref{fig:fingerprintgadget}.}
	\label{fig:labelgad}
\end{figure}
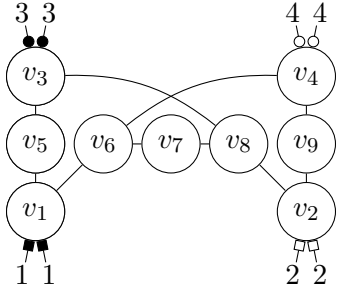

When we use several label gadgets simultaneously, there will be several labelings and we say an $e$ edge has label $l$ with respect to $v$ if $\lambda_v(e)=l$. We will now show how to replace a label gadget in a graph $G$ with a certain graph to obtain $G'$ such that the number of Hamiltonian cycles in $G'$ equals the number of Hamiltonian cycles in $G$ consistent with $(v,\lambda_v)$. The graph is shown at the left-hand side in Figure~\ref{fig:labelgad}. That is, the vertex $v$ is replaced by the displayed graph on vertices $v_1,\ldots,v_9$ and each edge with label $i=1,\ldots,4$ is connected to $v_i$. 

Any Hamiltonian cycle contains exactly two edges of the set of edges leaving the gadget and these have either labels $1$ and $2$ or labels $3$ and $4$.
This follows from a simple case analysis: If the cycle enters the gadget in vertex $v_1$, it must continue with $v_5,v_3$. Then it cannot leave the gadget, because then it is impossible to visit all six remaining vertices. Hence it must continue with $v_8$ and then $v_7,v_6,v_4,v_9$, and $v_2$ are forced and the cycle uses edges labeled with $1$ and $2$. The cases where it enters at a different vertex are symmetric.

\paragraph{Fingerprint Gadget.}

Now we present a general gadget allowing strong control on the fingerprints of partial solutions. If $B$ is a set of vertices, we let $\cP_B$ denote the set of all HC-fingerprints on $B$.

\begin{defn}[Fingerprint Gadget]
	A \emph{fingerprint gadget with boundary $B$} for a positive integer sequence $\{m_{f}\}_{f \in \cP_B}$ is a graph $G'=(V',E')$ such that $B \subseteq V$ and for every HC-fingerprint $f \in \cP_B$ the number of partial solutions in $G'$ for $f$ is exactly $m_f$.
\end{defn}

\begin{lem}[Fingerprint Gadget Implementation] \label{lem:finimp}
Let $G,B,\{m_f\}_{f \in \cP_B}$ as above. Assume there exist at least $2$ distinct fingerprints $f$ such that $m_{f}\neq 0$, and there exist $a,b \in B$ such that if $f=(d,M)$ and $m_f \neq 0$, we have $ab \in M$.\footnote{These are just technical conditions to facilitate the implementation of this gadget.} There is a fingerprint gadget $G'$ for $\{m_f\}_{f \in \cP_B}$ on $O(|B|\sum_{f \in \cP_B} m_f)$ vertices with path decomposition of width $|B|+O(1)$ that starts and ends in $B$.
\end{lem}
\begin{proof}
	We start with a formal definition of the graph $G'=(V',E')$, which we define using label gadgets. An illustrative example is provided in Figure~\ref{fig:fingerprintgadget}.
	\begin{enumerate}
		\item Add to $V'$ the set $B$ and additionally two vertices $a'$ and $c$. Add a subdivided edge $aa'$.\footnote{Formally, add another vertex $a''$ with only $a$ and $a'$ as neighbors. Here $a''$ is a vertex of degree $2$ with the sole purpose of enforcing the edges $a'a''$ and $a''a'$ to be in a Hamiltonian cycle.}
		\item Let $(f_i=(d_i,M_i))^\ell_{i=1}$ be a sequence of fingerprints that contains each $f \in \cP_B$ exactly $m_f$ times, so $\ell = \sum_{f \in \cP_B}m_f$.
		\item For $i=1,\ldots,\ell$:
		\begin{enumerate}
			\item Let $d^{-1}_i(2) = \{t^1_i,\ldots,t^{k}_i\}$ be the vertices of degree $2$ in $d_i$.
			\item Define a sequence of edges $E_i=(e^j_i)_{j=i}^{\ell_i}$ to be an arbitrary ordering of the union of the edgeset of a path from $a'$ through $d^{-1}_i(2)$ to $c$ and the edgeset $M_i \setminus ab$.
			\begin{itemize}
				\item  More formally, let $e^1_i= a't^1_i$; for $j=1,\ldots,k-1=|d^{-1}_i(2)|-1$ let $e^{j+1}_i=t^{j}_it^{j+1}_i$; let $e^{k+1}_i=t^{k}_ic$, and let $\{e^{k+2}_i,\ldots, e^{\ell_i}_i\}=M_i \setminus ab$, where $\ell_i = k+1+|M_i|$.
			\end{itemize}
			\item For $j=1,\ldots,\ell_i$:
			\begin{itemize}
				\item Let $e^j_i=uv$. Add a label gadget $p^j_i$ with incident edges $p^j_iu$ with label $1$ and $p^j_iv$ with label $2$.
				\item If $j>1$, add an edge $p_i^{j-1}p_i^j$ with $\lambda_{p_i^{j-1}}(e)=4$ and $\lambda_{p_i^{j}}(e)=3$.
			\end{itemize}
			\item If $i > 1$, add an edge $e=uv=p^{\ell_{i-1}}_{i-1}p^{1}_{i}$ with $\lambda_{u}(e)=4$ and $\lambda_{v}(e)=3$.
			\item If $i > 2$, add an edge $e=uv=p^{\ell_{i-2}}_{i-2}p^{1}_{i}$ with $\lambda_{u}(e)=4$ and $\lambda_{v}(e)=3$.
			\end{enumerate}
		\item For $i=1,2$, add an edge $e=uv=cp^i_1$ with $\lambda_{v}(e)=3$.
		\item For $i=\ell-1,\ell$, add an edge $e=uv=p^{\ell_i}_ib$ with $\lambda_{u}(e)=4$.
	\end{enumerate}

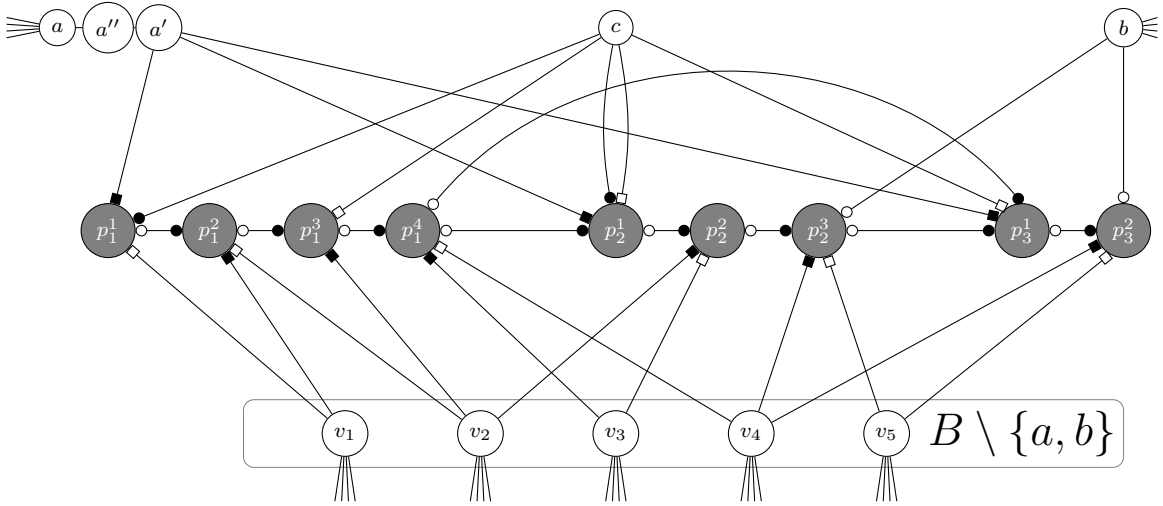
\begin{figure}
	\begin{center}
		\begin{tikzpicture}[scale=0.45]

\tikzstyle{bn}=[circle,fill=gray,text=white,draw=black]
\tikzstyle{wn}=[circle,fill=white,draw]
\tikzstyle{big}=[circle,minimum size=5cm, draw]
\tikzstyle{mythick}=[line width=4]
\tikzstyle{mydashed}=[dashed, line width=4]
\tikzstyle{lab1}=[{Square[scale=1.2]}-]
\tikzstyle{lab2}=[{Square[scale=1.2,open]}-]
\tikzstyle{lab3}=[{Circle[scale=1.2]}-]
\tikzstyle{lab4}=[{Circle[scale=1.2,open]}-]

\tikzstyle{lab43}=[{Circle[scale=1.2,open]}-{Circle[scale=1.2]}]

\small{

\scriptsize{
	\foreach \i in {1, 2, 3, 4}
	{
	  \node[bn] (p1\i) at (3 * \i-18,0) {$p_1^\i$};
	}
	
	\foreach \i in {1,2,3}
	{
	  \node[bn] (p2\i) at (3 * \i-3,0) {$p_2^\i$};
	}

	\foreach \i in {1,2}
	{
		\node[bn] (p3\i) at (3*\i +9,0) {$p_3^\i$};
	}
}

\draw[lab43] (p11) -- (p12);
\draw[lab43] (p12) -- (p13);
\draw[lab43] (p13) -- (p14);
\draw[lab43] (p14) -- (p21);

\draw[lab43] (p14) to[bend left=50] (p31.90);

\draw[lab43] (p21) -- (p22);
\draw[lab43] (p22) -- (p23);
\draw[lab43] (p23) -- (p31);
\draw[lab43] (p31) -- (p32);

\draw[rounded corners,line width=0.1,gray] (-11,-7) rectangle (15,-5);
\draw (12,-6) node {{\LARGE $B \setminus  \{a,b\}$}};

\node[wn] (v1) at (-8,-6) {$v_1$};
\node[wn] (v2) at (-4,-6) {$v_2$};
\node[wn] (v3) at (-0,-6) {$v_3$};
\node[wn] (v4) at (4,-6) {$v_4$};
\node[wn] (v5) at (8,-6) {$v_5$};

\node[wn] (a) at (-16.5,6) {$a$};
\node[wn] (app) at (-15,6) {$a''$};
\node[wn] (ap) at (-13.5,6) {$a'$};
\draw (a) -- (app) -- (ap);
\node[wn] (c) at (0,6) {$c$};
\node[wn] (b) at (15,6) {$b$};

\draw[lab1] (p11) -- (ap);
\draw[lab2] (p11) -- (v1);
\draw[lab1] (p12) -- (v1);
\draw[lab2] (p12) -- (v2);
\draw[lab1] (p13) -- (v2);
\draw[lab2] (p13) -- (c);
\draw[lab1] (p14.300) -- (v3);
\draw[lab2] (p14) -- (v4);

\draw[lab1] (p22) -- (v2);
\draw[lab2] (p22) -- (v3);
\draw[lab1] (p23) -- (v4);
\draw[lab2] (p23) -- (v5);
\draw[lab1] (p32) -- (v4);
\draw[lab2] (p32.240)  -- (v5);

\draw[lab1] (p21)--(ap);
\draw[lab1] (p31.150) -- (ap);

\draw[lab2] (p21) to[bend right=10] (c);
\draw[lab2] (p31.125) -- (c);

\draw[lab3] (p11) -- (c);
\draw[lab3]  (p21) to[bend left=10] (c);

\draw[lab4] (p23) -- (b);
\draw[lab4] (p32) -- (b);

\foreach \j in {-0.3,-0.1,0.1,0.3} 
{
	\draw (a) -- (-18,6+\j);
	\draw (b) -- (16,6+\j);
}

\foreach \i in {1,2,3,4,5}
\foreach \j in {-0.15cm,-0.05cm,0.05cm,0.15cm} 
{
	\draw let \p1=(v\i) in (2*\j+\x1,-8) -- (v\i);
}
}

\end{tikzpicture}
	\end{center}
	\caption{Example of the fingerprint gadget. Here $B=\{v_1,\ldots,v_5,a,b\}$ and $m_f$ equals one for $f=f_i=(d_i,M_i)$ (described in the next line), and zero otherwise. The vertices $p^1_i$ represent fingerprint $f_i$. Here $d^{-1}_1(2)=\{v_1,v_2\}$, $M_1=\{ab,v_3v_4\}$; $d^{-1}_2(2)=\emptyset$, $M_2=\{ab,v_2v_3,v_4v_5\}$; $d^{-1}_3(2)=\emptyset$, $M_3=\{ab,v_4v_5\}$. The vertices $p^j_i$ are label gadgets, and the symbols denote the labels $1,2,3,4$ as in Figure~\ref{fig:labelgad}.}
	\label{fig:fingerprintgadget}
\end{figure}

\paragraph{Correctness.} Let $H$ be a partial solution in $G'$ consistent with all label gadgets. Note that $H$ contains the edges $aa{''}$ and $a{''}a'$ as $a{''}$ has degree $2$. Thus we see that $a'$ needs to be adjacent in $H$ to $p^1_x$ for exactly one chosen $x$, as this vertex has no other neighbors. We claim that for every $x$ there is exactly one partial solution containing the edge $a'p^1_x$ and this partial solution has fingerprint $f_x$. Note that this is sufficient to prove the lemma, as it implies we have $m_f$ choices for $x$ that lead to fingerprint $f$.

To this end, suppose that $H$ contains the edge $a'p^1_x$, and let $i \neq x$. Then $H$ cannot contain an edge $e$ incident to $p^1_i$ satisfying $\lambda_{p^1_i}(e) \in \{1,2\}$ by the definition of a label gadget as it only has one edge not incident to $a'$ with such a label. Thus $H$ needs to contain edges incident to $p^1_i$ with labels $3$ and $4$ with respect to $p^1_i$. But the only edge with label $4$ is to $p^2_i$ (if it exists) which has label $3$ with respect to $p^2_i$. By propagation it follows that in $H$ we have that for every $p^j_i$ the two edges incident to $p^j_i$ must have label $3$ and $4$ with respect to $p^j_i$.

Now we focus on $p^1_x,\ldots,p^{\ell_x}_x$. We see that $p^2_x$ (if it exists) has only one incident edge with label $3$ (with respect to $p^2_x$) which is to $p^1_x$. Therefore the edges of $H$ incident to $p^2_x$ must have labels $1$ and $2$, and the same holds by propagation for all $p^j_x$. It follows that $H$ has fingerprint $f_x$: every vertex in $d^{-1}_x(2)$ has indeed degree $2$ as it is incident to two edges in the edges created in Step 3c from the set $E_x$; no edges incident to vertices in $d^{-1}_x(0)$ occur in $H$ as they do not occur in $E_x$ so they are not adjacent to the vertices $p^i_x$, which are the only vertices with incident edges with label $1$ or $2$. Moreover, for all edges in $uv \in M_i \setminus ab$ we see $H$ has the path $up^j_x,p^j_xv$ as the edges incident to $p^j_x$ must have labels $1$ and $2$.

Summarizing, we saw that $H$ must contain the paths $p^1_i,\ldots,p^{\ell_i}_i$ for $i \neq x$ and all edges with labels $1$ and $2$ with respect to vertices $p^j_x$, and if it does it has the correct fingerprint if $a$ and $b$ are connected to each other. It remains to show that (without creating subcycles) the paths can be connected to one path from $a$ to $b$ visiting all vertices $V' \setminus (B \cup \{p^{1+|d^{-1}_x(2)|}_x,\ldots,p^{\ell_i}_x)$ in a unique way. To see that this is the case, first note that $a'$ is connected to $c$ via the edges $e^1_x,\ldots,e^{k+1}_x \in E_i$ used in step 3b of the construction of $G'$ where $k=|d^{-1}_x(2)|$, so they are also connected in $G'$ via the vertices $p^1_x,\ldots,p^{k+1}_x$. To connect the paths $p^1_i,\ldots,p^{\ell_i}_i$ for $i \neq x$, note that $p^{\ell_i}_i$ can only be connected to vertices $p^j_{i'}$ with $i' >i$ as an incident edge of label $4$ must be chosen. It follows that the only way to complete the paths is to connect
\begin{itemize}
	\item  $c$ to $p^1_1$ if $x \neq 1$ or to $p^1_2$ if $x=1$,
	\item  $p^{\ell_i}_i$ to
	\begin{itemize}
		\item $p^1_{i+1}$, if $x \neq i+1$ and $i < \ell$,
		\item $p^1_{i+2}$, if $x=i+1$ and $i < \ell-1$,
		\item $b$, if $x=i+1$ and $i = \ell-1$,
		\item to $b$ if $i=\ell$,
	\end{itemize}
\end{itemize}

which connects $a$ to $b$, as required.

\paragraph{Path Decomposition and Size.} The claimed size bound holds trivially, as all label gadgets have constant size. For the pathwidth bound, note that after removing $B \cup \{a,b\}$ the graph induced on the vertices $p^j_i$ has constant pathwidth, as for every $i$, we have that $\{p^{\ell_i}_i,p^1_{i+1}\}$ forms a separator separating $p^{j}_{i'}$ with $i'\leq i$ from  $p^{j}_{i'}$ with $i > i'$, and the graph between separators $\{p^{\ell_{i-1}}_i,p^1_{i}\}$ and $\{p^{\ell_{i}}_i,p^1_{i+1}\}$ has constant pathwidth, as it is a path of label gadgets, each of constant size. The required path decomposition can thus be obtained by including $B$ in every bag.

\end{proof}

\subsection{The base case of Lemma~\texorpdfstring{\ref{lem:collem}}{5.2} }\label{sec:lemcollembase}

We prove Lemma~\ref{lem:collem} for $m=1$, so the CNF-formula is a single clause $C_1$.
The graph output by the reduction will consist of a graph $G_{i}$ for $1 \leq i \leq q$ (recall $q=n /\gamma$). Each $G_i$ contains `left boundary vertices' $L_{i}$, and `right boundary vertices' $R_{i}$, and additionally a `top vertex' $t_{i}$ and `bottom vertex' $b_{i}$. The graphs $G_i$ are glued together by unifying $b_{i}=t_{i-1}$ to get the graph $G$. The vertices $t_i$ and $b_i$ are used to propagate whether an encoded partial assignment has satisfied the clause. We now define the graph $G_i$. Recall from Subsection~\ref{subsec:patternpropagation} that $\eta_i$ is an injective function from $\cB_\rss$ to $\{0,1\}^{X_i}$.
\begin{defn}
	Let $b_{i},t_{i}$ be two vertices. The graph $G_i$ is the following instantiation of a fingerprint gadget as implemented by Lemma~\ref{lem:finimp} with boundary $B=L \cup R \cup \{b,t\}$, where we shorthand $L=L_i,R=R_i,t=t_i,b=b_i$.
	For every fingerprint $f_L=(d_L,M_L) \in \cB_\rss$ on $L$, fingerprint $f_R=(d_R,M_R) \in \cB_\lss$ on $R$ and $d_b,d_t \in \{0,2\}$, denote $f= (d_L \cup d_R \cup d_b \cup d_t, M^i_{alt})$, where $M_{alt}^i = M^i_L \cup M^i_S \setminus \{l_{i,1}l_{i,2},r_{i,1}r_{i,3}\} \cup \{l_{i,1}r_{i,1},l_{i,2}r_{i,3}\}$. For each such combinations we define $m_f=\FR^{-1}[f_L,f_R]$ if at least one of the following conditions holds:
	\begin{enumerate}
		\item $\eta_i(f_L)$ is an assignment of $X_i$ satisfying clause $C_1$, and
		\begin{enumerate}
			\item $d_t=d_b=2$, or
			\item both $d_b=2$ and $d_t=0$ hold.
		\end{enumerate}
		\item $\eta_i(f_L)$ is an assignment of $X_i$ not satisfying clause $C_1$, and
		\begin{enumerate}
			\item exactly one the propositions $d_t=2$ and $d_b=2$ holds.
		\end{enumerate}
	\end{enumerate}
	For other $f \in \cP(B)$, set $m_f=0$. Let $G_i$ be the fingerprint gadget with boundary $B$ for $(m_f)_{f \in \cP_B}$ .
\end{defn}

\begin{proof}[Proof of Lemma~\ref{lem:collem} (base case)]
We claim that the graph $G$ has the properties of Lemma~\ref{lem:collem}. Let $H$ be a partial solution which has fingerprint $f^i_L$ on $L_i$ and $f^i_R$ on $R_i$. Suppose that for every $i=1,\ldots,q$, the assignment $\eta_i(f^i_L)$ to $X_i$ does not satisfy $C_1$. Then exactly $q$ of the vertices $t_1,\ldots,t_q,b_q$ will have $2$ edges from $H$ incident to it so one vertex will have no incident edges in $H$ and therefore  $H$ cannot be a partial solution and does not contribute to the count.

Otherwise, let $i$ be the smallest integer such that the assignment $\eta_i(f^i_L)$ to $X_i$ satisfies $C_1$. It follows that in fingerprint gadget $G_i$, $H$ induces a partial solution for a fingerprint in which $t_i$ and $b_i$ both have degree $2$. Therefore, in the fingerprint gadget $G_{i'}$ with $i'<i$, $H$ induces a partial solution for a fingerprint in which $t_i$ has degree $2$ and $b_i$ has degree $0$. Similarly in the fingerprint gadget $G_{i'}$ with $i'>i$, $H$ induces a partial solution for a fingerprint in which $t_i$ has degree $0$ and $b_i$ has degree $2$. Therefore, in this case the number of combinations of partial solutions of $G_1,\ldots,G_q$ with the combined fingerprint is thus $(\FR^{-1})^{\otimes q}[f_L,f_R] = (\FR^{\otimes q})^{-1}[f_L,f_R]$ by the constructions of the graphs $G_i$, as required.

\paragraph{Pathwidth and Size.} Let $L_0=R_{q+1}=\emptyset$ for notational convenience. For $a=0,\ldots,q+1$, define
\[
	S_{< a} = \bigcup_{i=1}^{a-1} L_i \cup R_i \cup \{b_i,t_i\},\quad S_a= b_a \cup t_a \cup \bigcup_{i=a}^q L_i \cup \bigcup_{i=1}^a R_i, \quad S_{> a} = \bigcup_{i=a+1}^{q} L_i \cup R_i \cup \{b_i,t_i\}.
\]
Note that $S_a$ is a separator separating all vertices from $S_{<a}$ from all vertices from $S_{>a}$
Moreover, $S_0$ contains $L$ and $S_{q+1}$ contains $R$. 
We construct a path decomposition containing the bags $S_{0},\ldots,S_{q+1}$ in this order.
It remains to show that this can be completed into a efficient path decomposition by adding an appropriate path decompositions between bags $S_a$ and $S_{a+1}$.
To see this note that the vertices not in $S_{<a}$ and $S_{>a}$ must be in $G_a$ and $G_a$ admits a path decomposition of width $|B|+O(1)=\beta+O(1)$ starting in $L_a$ and ending in $R_a$ by Lemma~\ref{lem:finimp}.
Thus in between bags $S_a$ and $S_{a+1}$ we can add bags with $S_a$ and the path decomposition of $G_a$; after the last bag of this path decomposition we can forget all vertices of $G_a$ except $b_a$ and $R_a$ which are contained in $S_{a+1}$. We obtain a path decomposition of width $q\beta+O(\beta)$, as required.
\end{proof}

\subsection{Putting things together to prove Lemma~\texorpdfstring{\ref{modbound}}{5.1}}\label{sec:puttog}

\begin{proof}[Proof of Lemma~\ref{modbound}]
We first finish off the construction of $G$. Let $G'=G(L,R,\varphi)$ as in Lemma~\ref{lem:collem}. Then do the following for $i=1,\ldots,q$:
\begin{enumerate}
	\item Add a fingerprint gadget $G^R_i$ with boundary $R_i$ to $G'$ that has one partial solution for every fingerprint from $\cB_\rss$ on $R_i$.
	\item Add vertices $b_i,t_i$ to $G'$, where $b_i=t_{i+1}$ for $i < q$ and $t_q=b_1$.
	\item Add a fingerprint gadget $G^L_i$ with boundary $L_i \cup b_i \cup t_i$ to $G'$ such that for every fingerprint $f^a_i=(d,M) \in\cB_\lss$ on $L_i$, $G'$ has $\sum_{f^l \in \cB_\rss}(\FR^{\otimes q})^{-1}[f^l,f^a]$ partial solutions for the fingerpint
	\[
		f= (d',M \setminus \{l_{i,1}l_{i,3}\} \cup \{l_{i,1}t_{i},l_{i,3}b_i\}),
	\]
	on $L_i \cup b_i \cup t_i$. Here $d'$ equals $d$ with the addition that $d'(b_i)=d'(t_i)=1$.
\end{enumerate}

\paragraph{Number of solutions equals number of Hamiltonian cycles.}
Analogously to Claim~\ref{clm:psnr}, we first show the following:

\begin{claim}
The number of Hamiltonian cycles of $G$ equals
\begin{equation}\label{eq:sum}
\sum_{f^a,f^c \in \cB^{q}_\lss,f^l,f^b,f^d \in \cB^{q}_\rss} (\FR^{\otimes q})^{-1}[f^l,f^a] \FR^{\otimes q}[f^a,f^b] \mathbf{A}[f^b,f^c] \FR^{\otimes q}[f^c,f^d].
\end{equation}
\end{claim}
\begin{proof}
Denote $G^L$ and $G^R$ for the union of the graphs $G^L_i$ and $G^R_i$, respectively. Let $H$ be an edgeset with fingerprint constructed from $f^a=(f^a_1,\ldots,f^a_q)$ in Step 3 on $L$ in $G^L$, and fingerprint $f_L$ on $L$ in $G'$. By the construction of $G'$, the fingerprint $f_L$ must be obtained from a fingerprint $f_a$ by altering the matchings such that $f_a$ and $f_a$ match (indeed, as before, otherwise in one block two cycles will be created and one of them will contain the edge $\{l_{i,1},l_{i,3}\}$ and thus not be altered so it remains a subcycle). Similarly, if $H$ has fingerprint $f_d$ on $R$ in $G^R$ and fingerprint $f_R$ on $R$ in $G'$, $f_R$ must be obtained from a fingerprint $f_c$ by altering the matchings such that $f_c$ and $f_d$ match.

Conversely we claim that, if $H$ is an edgeset with fingerprints constructed from $f^a$ in Step 3 on $L$ in $G^L$, $f_L$ on $L$ in $G'$, $f_d$ on $R$ in $G^R$ and a fingerprint obtained from $f^c$ on $R$ in $G'$ by altering the matching, and both $f^a$ and $f^b$ as well as $f^c$ and $f^d$ match, then $H$ is automatically a Hamiltonian cycle. To see this, first take into account the partial solution in $G'$ and $G^R$. Similarly as in the proof of Claim~\ref{clm:psnr}, it is easily seen that this gives a set of paths that connect $l_{i,1}$ with $l_{i,2}$ for every $i$ as the fingerprints $f^c$ and $f^d$ match. Taking also the partial solutions in $G^L$ into account and that $f^a$ and $f^b$ match, we see that within each block $G^i_L$ the partial solutions connect $t_i$ to $b_i$ (denoting $t_i=b_q$) and therefore the partial solutions combined give a Hamiltonian cycle.

By the way we constructed $G^L_i$ the number of partial solutions that have fingerprint $f^a$ on $L$ is multiplied with $\sum_{f^l \in \cB_\rss}(\FR^{\otimes q})^{-1}[f^l,f^a]$.
Therefore, by summing over all fingerprints $f^a,f^b,f^c,f^d$ and counting number of partial solutions with these fingerprints as described in the fingerprint gadgets we obtain that the number of Hamiltonian cycles equals the claimed quantity.
\end{proof}

By the definition of $\mathbf{A}$ from Lemma~\ref{lem:collem}, $\mathbf{A}=\mathbf{C}(\FR^{\otimes q})^{-1}$ with 
\[ \mathbf{C}[f,f] = 
\begin{cases}
1, & \text{if } \eta(f) \models \varphi,\\
0, &\text{otherwise}.
\end{cases}
\]
Thus, in matrix terms,~\eqref{eq:sum} can be rewritten into
\[
	\mathbf{1}^T(\FR^{\otimes q})^{-1}\FR^{\otimes q} \mathbf{C}(\FR^{\otimes q})^{-1}\FR^{\otimes q}\mathbf{1} = \mathbf{1}^T\mathbf{C}\mathbf{1},
\]
which is easily seen to be the number of assignments of $X$ that satisfy $\varphi$ modulo $p$, as required.

\paragraph{Pathwidth bound.} Recall the graph $G'$ has pathwidth $q\beta+O(\beta)$. It is easy to see that the additions of the graphs $G^L_i$ and $G^R_i$ do not increase the pathwidth beyond this bound: We can simply  introduce and forget each $G^R_i$ separately at the end of the path decomposition. As similar approach can be used for $G^L_i$ in the start of the path decomposition where we each time only forget the top vertex (except $t_1$). Recall that $q=n/\gamma$. We will now set the parameters $\beta$ and $\gamma$. We first show that we can find the needed sufficiently large sets $\cB_\lss$ and $\cB_\rss$:
\begin{lem}\label{clm:basislb}
	Let $\rk_p \in \R$ be such that $\log_{\rk_p}(\rank(\MC_k))/k \to c$, where $c \geq 1$, as even $k$ tends to infinity.
	Then for any $\eps'$, there exists a large enough $\beta$ and sets $\cB_\lss,\cB_\rss$ of fingerprints on $[\beta]$ of size at least $(2+\rk_p-\eps')^{\beta}$ such that $\HC_\beta[\cB_\lss,\cB_\rss]$ has full rank over $p$ and if $(d_L,M_L) \in \cB_\lss$ and $(d_R,M_R) \in \cB_\rss$ then
	\begin{inparaenum}[(i)]
		\item $d_\lss(1), d_\rss(1), d_\lss(2)$ and $d_\rss(3)$ all equal $1$, and
		\item $\{1,3\} \in M_\lss$ and $\{1,2\} \in M_\rss$.
	\end{inparaenum}
\end{lem}
\begin{proof}
We have $\rank(\HC_\beta) \geq (2+r_p-\eps')^\beta$ for large enough $\beta$ by Fact~\ref{fact:rank-MC-to-HC}, and by the same binomial theorem argument as in Fact~\ref{fact:rank-MC-to-HC} we also have that sets of linearly independent rows and columns exist consisting of fingerprints $(d,M)$ satisfying $|d^{-1}(1)| \geq 4$. Then there must be a vertex matched to the same vertices in the row index $M_1$ and column index $M_2$ in at least an $1/\beta^2$ fraction of both basis matchings by the pigeon principle. Denoting this vertex and its two frequent neighbors with $1,2,3$ the claim follows.
\end{proof}

Let $\hat{\eps}=\eps/2$, and pick $\beta$ sufficiently large such that Lemma~\ref{clm:basislb} ensures the sets $\cB_\lss,\cB_\rss$ of fingerprints on $[\beta]$ of size at least $(2+\rk_p-\hat{\eps})^{\beta}$ exist. To ensure the existence of the injective encoding functions $\eta_1,\ldots,\eta_q$, we pick $\gamma$ such that $2^\gamma < (2+\rk_p-\hat{\eps})^{\beta}$; set $\gamma$ as large as possible under this constraint so that $\gamma \geq \beta \log_2(2+\rk_p-\hat{\eps})-1$. The pathwidth of $G$ will be at most 
\[
	q\beta+O(\beta)=(n/\gamma) \beta+O(\beta) \leq n/(\log_2(2+\rk_p-\hat{\eps}) -1/\beta)+O(\beta).
\]

The running time of the assumed algorithm for counting the number of Hamiltonian cycles modulo $p$ of the created instance will thus be $O^*(2^{\alpha n+O(\beta)})=O^*(2^{\alpha n})$ time where
\[
	\alpha = \frac{\log_2(2+\rk_p-\eps)}{\log_2(2+\rk_p-\hat{\eps})-1/\beta},
\]
which is smaller than $1$ for sufficiently large $\beta$ (which may depend on $\eps$). Finally, it can be easily checked that the graph $G$ can trivially be constructed in time polynomial in the size of $\varphi$ and $p$ for constant $\eps$.
\end{proof}

\section{Conclusions}
\label{sec: summary}

As future work, we suggest the problem of counting the connected induced subgraphs of a graph, where one could try to exclude $O^*(2^{o(\pw \log \pw)})$ time algorithms. 
The connection matrix for this problem is the meet matrix of the partition lattice, ordered by refinement, so that the coarsest partition with one block is the smallest element.
For this setting, the powerful technology of M\"obius functions (see e.g.~\cite{LOVASZ1993322}) can readily give rank lower bounds, 
but it is not a priori clear how to construct the gadgets required for converting the rank bound into an algorithmic lower bound.
Another example could be the problem of counting Steiner Trees, which has an $O^*(5^{\pw})$ time algorithm from~\cite{DBLP:journals/iandc/BodlaenderCKN15}, or the evaluation of graph polynomials such as the Tutte polynomial.

A further natural direction for future research is to find the optimal constant $c_p$ such that \CHCP modulo $p$ can be solved in $O^*(c_p^{\pw})$ time and not in $O^*((c_p-\varepsilon)^{\pw})$ time for $\varepsilon >0$. It is natural to conjecture that $c_p=2+r_p$, where $r_p$ is the exponential base of the rank of $\MC_k$ over $\Z_p$. However, note that in~\cite{CyganKN13}, obtaining an algorithm from the rank upper bound was not trivial and, unlike the lower bound from Theorem~\ref{thm:mainrank}, it is not a priori clear how to use rank \emph{upper bounds} as a black box. The main reason is that we cannot seem to reduce the work related to $\MC$ to constant-sized copies as done in the proof of Theorem~\ref{thm:mainrank}.
Additionally, we need better lower bounds for the rank of $\MC_k$ over $\Z_p$,
since the bounds from our paper are tight only over $\Q$.

\subsection*{Acknowledgments}
The authors thank anonymous reviewers for their valuable comments, and one reviewer in particular for his thoroughness. 
Jesper acknowledges Marek Cygan for discussions on Section~\ref{sec:red}.

\bibliographystyle{plain}
\bibliography{ms}

\begin{thebibliography}{10}

\bibitem{Alman17}
Josh Alman and Ryan Williams.
\newblock Probabilistic rank and matrix rigidity.
\newblock In {\em Proceedings of the 49th Annual ACM SIGACT Symposium on Theory
  of Computing}, STOC 2017, pages 641--652, New York, NY, USA, 2017. ACM.

\bibitem{Regev10}
Regev Amitai.
\newblock Identities for the number of standard {Y}oung tableaux in some
  $(k,l)$-hooks.
\newblock {\em S\'eminaire Lotharingien de Combinatoire}, 63(3), 2010.

\bibitem{BannaiI84}
Eiichi Bannai and Tatsuro Ito.
\newblock {\em Algebraic combinatorics I: association schemes}.
\newblock Mathematics lecture note series. Benjamin/Cummings Pub. Co., 1984.

\bibitem{DBLP:journals/siamcomp/Bjorklund14}
Andreas Bj{\"{o}}rklund.
\newblock Determinant sums for undirected hamiltonicity.
\newblock {\em {SIAM} J. Comput.}, 43(1):280--299, 2014.

\bibitem{DBLP:conf/focs/BjorklundH13}
Andreas Bj{\"{o}}rklund and Thore Husfeldt.
\newblock The parity of directed hamiltonian cycles.
\newblock In {\em 54th Annual {IEEE} Symposium on Foundations of Computer
  Science, {FOCS} 2013, 26-29 October, 2013, Berkeley, CA, {USA}}, pages
  727--735. {IEEE} Computer Society, 2013.

\bibitem{DBLP:conf/icalp/BjorklundKK17}
Andreas Bj{\"{o}}rklund, Petteri Kaski, and Ioannis Koutis.
\newblock Directed hamiltonicity and out-branchings via generalized laplacians.
\newblock In {\em 44th International Colloquium on Automata, Languages, and
  Programming, {ICALP} 2017, July 10-14, 2017, Warsaw, Poland}, pages
  91:1--91:14, 2017.

\bibitem{DBLP:journals/iandc/BodlaenderCKN15}
Hans~L. Bodlaender, Marek Cygan, Stefan Kratsch, and Jesper Nederlof.
\newblock Deterministic single exponential time algorithms for connectivity
  problems parameterized by treewidth.
\newblock {\em Inf. Comput.}, 243:86--111, 2015.

\bibitem{DBLP:books/daglib/0090316}
Peter B{\"{u}}rgisser, Michael Clausen, and Mohammad~Amin Shokrollahi.
\newblock {\em Algebraic complexity theory}, volume 315 of {\em Grundlehren der
  mathematischen Wissenschaften}.
\newblock Springer, 1997.

\bibitem{DBLP:journals/jcss/CalabroIKP08}
Chris Calabro, Russell Impagliazzo, Valentine Kabanets, and Ramamohan Paturi.
\newblock The complexity of unique $k$-{SAT}: An isolation lemma for
  $k$-{CNF}s.
\newblock {\em J. Comput. Syst. Sci.}, 74(3):386--393, 2008.

\bibitem{DBLP:conf/focs/CurticapeanX15}
Radu Curticapean and Mingji Xia.
\newblock Parameterizing the permanent: Genus, apices, minors, evaluation mod
  $2^k$.
\newblock In {\em {IEEE} 56th Annual Symposium on Foundations of Computer
  Science, {FOCS} 2015, Berkeley, CA, USA, 17-20 October, 2015}, pages
  994--1009, 2015.

\bibitem{CyganKN13}
Marek Cygan, Stefan Kratsch, and Jesper Nederlof.
\newblock Fast hamiltonicity checking via bases of perfect matchings.
\newblock In {\em Symposium on Theory of Computing Conference, STOC'13, Palo
  Alto, CA, USA, June 1-4, 2013}, pages 301--310, 2013.

\bibitem{DBLP:conf/focs/CyganNPPRW11}
Marek Cygan, Jesper Nederlof, Marcin Pilipczuk, Micha\l{} Pilipczuk, Johan
  M.~M. van Rooij, and Jakub~Onufry Wojtaszczyk.
\newblock Solving connectivity problems parameterized by treewidth in single
  exponential time.
\newblock In Rafail Ostrovsky, editor, {\em {IEEE} 52nd Annual Symposium on
  Foundations of Computer Science, {FOCS} 2011, Palm Springs, CA, USA, October
  22-25, 2011}, pages 150--159. {IEEE} Computer Society, 2011.

\bibitem{DBLP:journals/jacm/FominLPS16}
Fedor~V. Fomin, Daniel Lokshtanov, Fahad Panolan, and Saket Saurabh.
\newblock Efficient computation of representative families with applications in
  parameterized and exact algorithms.
\newblock {\em J. {ACM}}, 63(4):29:1--29:60, 2016.

\bibitem{GodsilMeagher}
Chris Godsil and K.~Meagher.
\newblock {\em Erdos-Ko-Rado Theorems: Algebraic Approaches}.
\newblock Cambridge Studies in Advanced Mathematics. Cambridge University
  Press, 2015.

\bibitem{GodsilRoyle}
Chris Godsil and Gordon Royle.
\newblock {\em Algebraic graph theory}, volume 207 of {\em Graduate Texts in
  Mathematics}.
\newblock Springer-Verlag, New York, 2001.

\bibitem{DBLP:conf/mfcs/1977}
Jozef Gruska, editor.
\newblock {\em Mathematical Foundations of Computer Science 1977, 6th
  Symposium, Tatranska Lomnica, Czechoslovakia, September 5-9, 1977,
  Proceedings}, volume~53 of {\em Lecture Notes in Computer Science}. Springer,
  1977.

\bibitem{HardyWright08}
Godfrey~H. Hardy and Edward~M. Wright.
\newblock {\em An Introduction to the Theory of Numbers}.
\newblock Oxford University Press, Oxford, sixth edition, 2008.
\newblock Revised by D. R. Heath-Brown and J. H. Silverman, With a foreword by
  Andrew Wiles.

\bibitem{held1962dynamic}
Michael Held and Richard~M Karp.
\newblock A dynamic programming approach to sequencing problems.
\newblock {\em Journal of the Society for Industrial and Applied Mathematics},
  10(1):196--210, 1962.

\bibitem{ImpagliazzoP01}
Russel Impagliazzo and Ramamohan Paturi.
\newblock On the complexity of $k$-{SAT}.
\newblock {\em J. Comput. Syst. Sci.}, 62(2):367--375, 2001.

\bibitem{DBLP:conf/csl/KotekM12}
Tomer Kotek and Johann~A. Makowsky.
\newblock Connection matrices and the definability of graph parameters.
\newblock In {\em Computer Science Logic (CSL'12) - 26th International
  Workshop/21st Annual Conference of the EACSL, {CSL} 2012, September 3-6,
  2012, Fontainebleau, France}, pages 411--425, 2012.

\bibitem{DBLP:journals/talg/KratschW14}
Stefan Kratsch and Magnus Wahlstr{\"{o}}m.
\newblock Compression via matroids: {A} randomized polynomial kernel for odd
  cycle transversal.
\newblock {\em {ACM} Trans. Algorithms}, 10(4):20:1--20:15, 2014.

\bibitem{Kushilevitz:1996:CC:264772}
Eyal Kushilevitz and Noam Nisan.
\newblock {\em Communication Complexity}.
\newblock Cambridge University Press, New York, NY, USA, 1997.

\bibitem{LINDZEY2017130}
Nathan Lindzey.
\newblock Erdős–{K}o–{R}ado for perfect matchings.
\newblock {\em European Journal of Combinatorics}, 65:130 -- 142, 2017.

\bibitem{DBLP:conf/soda/LokshtanovMS11a}
Daniel Lokshtanov, D{\'{a}}niel Marx, and Saket Saurabh.
\newblock Known algorithms on graphs on bounded treewidth are probably optimal.
\newblock In Dana Randall, editor, {\em Proceedings of the Twenty-Second Annual
  {ACM-SIAM} Symposium on Discrete Algorithms, {SODA} 2011, San Francisco,
  California, USA, January 23-25, 2011}, pages 777--789. {SIAM}, 2011.

\bibitem{DBLP:journals/ejc/Lovasz06}
L{\'{a}}szl{\'{o}} Lov{\'{a}}sz.
\newblock The rank of connection matrices and the dimension of graph algebras.
\newblock {\em Eur. J. Comb.}, 27(6):962--970, 2006.

\bibitem{DBLP:books/daglib/0031021}
L{\'{a}}szl{\'{o}} Lov{\'{a}}sz.
\newblock {\em Large Networks and Graph Limits}, volume~60 of {\em Colloquium
  Publications}.
\newblock American Mathematical Society, 2012.

\bibitem{DBLP:conf/focs/LovaszS88}
L{\'{a}}szl{\'{o}} Lov{\'{a}}sz and Michael~E. Saks.
\newblock Lattices, {M\"{o}}bius functions and communication complexity.
\newblock In {\em 29th Annual Symposium on Foundations of Computer Science,
  White Plains, New York, USA, 24-26 October 1988}, pages 81--90. {IEEE}
  Computer Society, 1988.

\bibitem{LOVASZ1993322}
László Lovăsz and Michael Saks.
\newblock Communication complexity and combinatorial lattice theory.
\newblock {\em Journal of Computer and System Sciences}, 47(2):322 -- 349,
  1993.

\bibitem{MacDonald95}
Ian~G. Macdonald.
\newblock {\em Symmetric functions and Hall polynomials}.
\newblock Oxford mathematical monographs. Clarendon Press, 1995.

\bibitem{matouvsek2010thirty}
Ji{\v{r}}{\'\i} Matou{\v{s}}ek.
\newblock {\em Thirty-three miniatures: Mathematical and Algorithmic
  applications of Linear Algebra}, volume~53.
\newblock American Mathematical Soc., 2010.

\bibitem{Monien1985}
Burkhard Monien.
\newblock How to find long paths efficiently.
\newblock {\em North-Holland Mathematics Studies}, 109:239 -- 254, 1985.
\newblock Analysis and Design of Algorithms for Combinatorial Problems.

\bibitem{A=B}
Marko Petkovsek, Herbert~S. Wilf, and Doron Zeilberger.
\newblock {\em A = B}.
\newblock A {K} Peters Series. Taylor \& Francis, 1996.

\bibitem{RazS95}
Ran Raz and Boris Spieker.
\newblock On the "log rank"-conjecture in communication complexity.
\newblock {\em Combinatorica}, 15(4):567--588, 1995.

\bibitem{Sagan}
Bruce Sagan.
\newblock {\em The Symmetric Group: Representations, Combinatorial Algorithms,
  and Symmetric Functions}.
\newblock Graduate Texts in Mathematics. Springer New York, 2001.

\bibitem{StanleyV201}
Richard~P. Stanley.
\newblock {\em Enumerative Combinatorics:}, volume~2 of {\em Cambridge Studies
  in Advanced Mathematics}.
\newblock Cambridge University Press, 2001.

\bibitem{Thrall42}
Robert~M. Thrall.
\newblock On symmetrized {K}ronecker powers and the structure of the free {L}ie
  ring.
\newblock {\em American Journal of Mathematics}, 64(1):pp. 371--388, 1942.

\bibitem{DBLP:conf/fsttcs/Williams14}
Ryan Williams.
\newblock The polynomial method in circuit complexity applied to algorithm
  design (invited talk).
\newblock In Venkatesh Raman and S.~P. Suresh, editors, {\em 34th International
  Conference on Foundation of Software Technology and Theoretical Computer
  Science, {FSTTCS} 2014, December 15-17, 2014, New Delhi, India}, volume~29 of
  {\em LIPIcs}, pages 47--60. Schloss Dagstuhl - Leibniz-Zentrum fuer
  Informatik, 2014.

\bibitem{DBLP:conf/iwpec/Wlodarczyk16}
Micha\l{} W\l{}odarczyk.
\newblock Clifford algebras meet tree decompositions.
\newblock In Jiong Guo and Danny Hermelin, editors, {\em 11th International
  Symposium on Parameterized and Exact Computation, {IPEC} 2016, August 24-26,
  2016, Aarhus, Denmark}, volume~63 of {\em LIPIcs}, pages 29:1--29:18. Schloss
  Dagstuhl - Leibniz-Zentrum fuer Informatik, 2016.

\end{thebibliography}

\appendix

\section{Proof of Theorem~\texorpdfstring{\ref{thm:mainlb}}{1.1} from Theorem~\texorpdfstring{\ref{thm:genlowhc}}{1.2} and Theorem~\texorpdfstring{\ref{thm:mainrank}}{1.3}}\label{sec:crt}

We prove a slightly stronger consequence, namely, that there is an algorithm that counts the satisfying assignments of a given a CNF-formula on $n$ variables in $O^*((2-\eps)^n)$ time for some $\eps>0$.

Let $n$ be the number of variables of the given CNF-formula $\varphi$. The Chinese Remainder Theorem (CRT) tells us that given the number of satisfying assignments of $\varphi$ modulo primes $p_1,\ldots,p_\ell$, we can compute the number of satisfying solutions of $\varphi$ as long as $\prod^\ell_{i=1}p_i \geq 2^n$. By the Prime Number Theorem~\cite[p.~494, Eq.~(22.19.3)]{HardyWright08}, there are at least $r/\log_2 r$ primes between $r$ and $2r$, and thus
\[
	\prod_{r\, \leq \, p \text{ prime} \, \leq \, 2r } p \geq r^{r / \log r} \geq 2^{\Omega(r)}.
\]
It follows that for counting the number of satisfying assignments of a given CNF-formula, it is sufficient to count the number of satisfying assignments modulo $p$ for any $p=\Theta(n)$.
We do this using Lemma~\ref{modbound} combined with the algorithm for \CHCP. For fixed $t$ we have that $\rank_p(\MC_t)=\rank(\MC_t)$ for large enough $p$ (which can for example be shown by upper bounding the determinant of $\MC_t$ by $t!$). The assumed algorithm for \CHCP also counts the number of Hamiltonian cycles modulo $p$. By Theorem~\ref{thm:mainrank} we have $\lim_{p \rightarrow \infty}r_p=4$ and Theorem~\ref{thm:mainlb} follows.
\section{Finite Group Representation Theory}\label{sec:repbackground}
 
For a set $X$, let $\mathbb{C}[X]$ denote the vector space of dimension $|X|$ of complex-valued functions from $X$.
A \emph{representation} $(\phi,V)$ of a finite group $G$ is a homomorphism $\phi : G \rightarrow \mathbb{GL}(V)$ where $\mathbb{GL}(V)$ is the group of $\dim V \times \dim V$ invertible matrices. We refer to $(\phi,V)$ simply as $\phi$ when $V$ is understood, or as $V$ when $\phi$ is understood.  For any representation $\phi$, we define its \emph{dimension} to be $\dim \phi := \dim V$. 
Two representations $\rho,\phi$ are \emph{equivalent} if $\rho(g)$ and $\phi(g)$ are similar for all $g \in G$.

Let $(\phi,V)$ be a representation of a finite group $G$, and let $W \leq V$ be a \emph{G-invariant} subspace, that is, $\phi(g)w \in W$ for all $w \in W$ and for all $g \in G$.  We say that $(\phi|_W,W)$ is a \emph{sub-representation} of $\phi$ where $\phi|_W$ is the restriction of $\phi$ to the subspace $W$. A representation $(\phi,V)$ is an \emph{irreducible representation} (or simply, an \emph{irreducible}) if it has no proper sub-representations. 

It is well-known that there is a one-to-one correspondence between the set of inequivalent irreducibles of $G$ and its conjugacy classes $\mathcal{C}$, and that any representation $V$ of $G$ uniquely decomposes as a finite direct sum of inequivalent irreducibles $V_i$ of $G$:
\[V \cong \bigoplus_{i=1}^{|\mathcal{C}|} \,m_i\,V_i\]
where $m_i$ is the \emph{multiplicity} of $V_i$ (the number of times that $V_i$ occurs in the decomposition).
Natural representations of groups can be obtained by letting them act on sets.
In particular, for any group $G$ acting on a set $X$, let $(\phi,\mathbb{C}[X])$ be the \emph{permutation representation} of $G$ on $X$ defined such that
\[\phi(g)[f(x)] = f(g^{-1}x)\]
for all $g \in G$, $f \in \mathbb{C}[X]$, and $x \in X$.
\noindent If we let $G$ act on itself $(X = G)$, then we obtain the \emph{regular representation}, which admits the following decomposition into irreducibles:
\[ \mathbb{C}[G] \cong \bigoplus_{i=1}^{|\mathcal{C}|}\, (\dim \phi_i) \,V_i\]
where $(\phi_i,V_i)$ is the $i$th irreducible of $G$.
Letting $e_ge_h = e_{gh}$ over the standard basis $\{e_g\}_{g \in G}$ of $\mathbb{C}[G]$, we see that $\mathbb{C}[G]$ is an algebra, the so-called \emph{group algebra} of $G$. 

For any (irreducible) representation $\phi$ of $G$, the \emph{(irreducible) character} $\chi_\phi$ of $\phi$ is the map $\chi_\phi : G \rightarrow \mathbb{C}$ such that $\chi_\phi (g) := \text{Tr}(\phi(g))$. Similar matrices have the same trace, thus the character of a representation is a \emph{class function}, that is, they are constant on conjugacy classes. Furthermore, the characters of the set of all irreducible representations of a group $G$ form an orthonormal basis for the space of all class functions of $\mathbb{C}[G]$.

\end{document}